%% file: main.tex
\documentclass[a4paper,UKenglish]{article}
\usepackage{microtype}
\usepackage{geometry}

\usepackage{amsmath}
\usepackage{amssymb,array,amsthm}

\usepackage{paralist}
\usepackage{tabularx}
\usepackage{changepage}
\usepackage{enumerate}
\usepackage{todonotes}
\usepackage{graphicx}
\usepackage{algorithm}
\usepackage[noend]{algorithmic}
\usepackage{tikz}
\usepackage{bm}
\usepackage{authblk}

\usetikzlibrary{arrows.meta,decorations.markings,circuits,decorations.pathmorphing,decorations.pathreplacing,arrows,automata,}

\usepackage{subcaption}

\input{macros}
\theoremstyle{plain}
\newtheorem{proposition}{Proposition}
\newtheorem{example}{Example}
\newtheorem{theorem}{Theorem}
\newtheorem{lemma}{Lemma}
\newtheorem{remark}{Remark}

\date{}

\title{Admissibility in Concurrent Games\footnote{Work partially
    supported by the ERC inVEST (279499) project.}}

\author[1]{Nicolas Basset}
\author[1]{Gilles Geeraerts}
\author[1]{Jean-Fran\c{c}ois Raskin}
\author[2]{Ocan Sankur}

\affil[1]{Universit\'e libre de Bruxelles, Brussels, Belgium\\
  \texttt{\{nicolas.basset, gilles.geeraerts, jraskin\}@ulb.ac.be}}
\affil[2]{CNRS, IRISA, Rennes, France\\
  \texttt{ocan.sankur@irisa.fr}}

\begin{document}

\maketitle 

\begin{abstract}
  In this paper, we study the notion of \emph{admissibility}
  for randomised strategies in
  \emph{concurrent games}. Intuitively, an admissible strategy is one
  where the player plays  `as well as possible', because
  there is no other strategy that \emph{dominates} it, i.e., that wins (almost surely)
  against a super set of adversarial strategies. We prove that admissible
  strategies always exist in concurrent games, and we characterise
  them precisely. Then, when the objectives of the players are
  $\omega$-regular, we show how to perform \emph{assume-admissible
    synthesis}, i.e., how to compute admissible strategies that win (almost surely)
  under the hypothesis that the other players play admissible
  strategies only.
\end{abstract}

\input{intro}

\input{prelim}

\input{adm}

\input{aasynth}

\input{main.bbl}

\end{document}

%% file: macros.tex
\newcommand{\cyl}[1]{\textrm{Cyl}(#1)}

\newcommand{\eqadm}[1]{\simeq_{#1}}
\newcommand{\botheqadm}[2]{\simeq^{#1}_{#2}}
\newcommand{\botheqadmstrat}[1]{\approx^{#1}}
\newcommand{\supp}{\textrm{Supp}}

\newcommand{\sure}{\texttt{S}}
\newcommand{\asure}{\texttt{A}}
\newcommand{\bothmodels}[2]{\models^{#1} #2}
\newcommand{\bothmodelsh}[3]{\models^{#1}_{#2} #3}
\newcommand{\rmodels}[1]{\bothmodels{\asure}{#1}}
\newcommand{\smodels}[1]{\bothmodels{\sure}{#1}}
\newcommand{\starmodels}[1]{\bothmodels{\star}{#1}}
\newcommand{\starmodelsh}[2]{\bothmodelsh{\star}{#1}{#2}}
\newcommand{\rmodelsh}[2]{\bothmodelsh{\asure}{#1}{#2}}
\newcommand{\smodelsh}[2]{\bothmodelsh{\sure}{#1}{#2}}

\newcommand{\distr}[1]{\mathcal{D}(#1)}
\newcommand{\Prob}{\mathbb{P}}

\newcommand{\bothvalpred}[3]{\texttt{Val}^{#1}_{#2,#3}}
\newcommand{\bothvalpredstates}[3]{\texttt{Val}^{#1}_{#2,#3}}
\newcommand{\rectset}{\Gamma}
\newcommand{\rectsetdet}{\Gamma^{\textrm{det}}}

\newcommand{\choicep}[1]{\Sigma_{#1}}
\newcommand{\transfun}{\delta}
\newcommand{\transfunrand}{\delta_{r}}
\newcommand{\transfunbis}{\overline{\delta}}
\newcommand{\vinitbis}{\overline{\ensuremath{s_{\mathsf{init}}}}}
\newcommand{\Sigmabis}{\overline{\Sigma}}
\newcommand{\choicepbis}[1]{\Sigmabis_{#1}}
\newcommand{\statesbis}{\overline{\states}}
\newcommand{\obs}{\mathfrak{o}}

\newcommand{\N}{\mathbb{N}}

\newcommand{\game}{\ensuremath{\mathcal{G}}}

\newcommand{\states}{S}

\newcommand{\Runs}{\texttt{Runs}}

\newcommand{\outc}{\texttt{Outcome}}
\newcommand{\prefixoutc}{\texttt{Hist}}
\newcommand{\outcome}[1]{\texttt{Outcome}(#1)}
\newcommand{\prefixoutcome}[1]{\texttt{Hist}(#1)}
\newcommand{\supphist}[1]{\texttt{Hist}(#1)}
\newcommand{\profile}{\ensuremath{{\vec \strat}}}

\newcommand{\strat}{\sigma}
\newcommand{\vinit}{\ensuremath{s_{\mathsf{init}}}}

\newcommand{\valhist}[2]{\chi^{#2}_{#1}}
\newcommand{\valhstrat}[3]{\chi_{#3}^{#1}(#2)}
\newcommand{\valh}[3]{\chi_{#3}^{#1}(#2)}

\newcommand{\win}[1]{\Phi(#1)}

\newcommand{\bothphivalone}[2]{\Phi_1^{#1}(#2)}
\newcommand{\bothphivalzero}[2]{\Phi_0^{#1}(#2)}

 \newcommand{\bothgameadmp}[2]{\game^{#1}_{#2}}

 \newcommand{\bothdomstrat}[1]{\prec^{#1}}
 \newcommand{\bothdomstrateq}[1]{\preccurlyeq^{#1}}

\newcommand{\bothdommove}[2]{<^{#1}_{#2}}
 \newcommand{\bothdommoveeq}[2]{\leqslant^{#1}_{#2}}

\newcommand{\hist}{h}

\newcommand{\last}[1]{\mathrm{last}(#1)}

\newcommand{\indexplayers}{P}
\newcommand{\players}{\indexplayers}

\newcommand{\bothPhiAA}[2]{\Phi_{\bothgameadmp{#1}{#2}}}

\newcommand{\prefix}{\ensuremath{\subseteq_{\mathsf{pref}}}}

\newcommand{\bothAfterHelp}[3]{\texttt{AfterHelp}^{#1}_{#2,#3}}
\newcommand{\bothAfterHelpMove}[2]{\texttt{AfterHelpMove}^{#1}_{#2}}

\newcommand{\succmove}[2]{\ensuremath{\mathrm{Succ}(#1,#2)}}
\newcommand{\succmovenew}[1]{\ensuremath{\mathrm{Succ}(#1)}}

\newcommand{\switchstrat}[3]{#1\langle#2\leftarrow #3\rangle}
\newcommand{\defeq}{\ensuremath{=_{\mathrm{def}}}}

\renewcommand{\vec}[1]{\ensuremath{\bm{#1}}}

\newcommand{\target}{\ensuremath{\mathsf{Trg}}}

\newcommand{\stratbis}{\ensuremath{\overline{\strat}}}
\newcommand{\stratter}{\ensuremath{\hat{\strat}}}
\newcommand{\histbis}{\ensuremath{\overline{\hist}}}
\newcommand{\stateofaction}{\ensuremath{\mathsf{st}}}

%% file: intro.tex
\section{Introduction}

In a concurrent $n$-player game played on a graph, all $n$ players
{\em independently} and {\em simultaneously} choose moves at each
round of the game, and those $n$ choices determine the next state of
the game~\cite{DBLP:journals/tcs/AlfaroHK07}. Concurrent games
generalise turn-based games and it is well-known that, while
deterministic strategies are sufficient in the turn-based case,
randomised strategies are necessary for winning with probability one
even for reachability objectives. Intuitively, randomisation is
necessary because in concurrent games, in each round, players have no
information about the concurrent choice of moves made by the other
players. Randomisation allows for some probability of choosing a good
move while not knowing the choice of the other players. As a
consequence, there are two classical semantics that are considered to
analyse these games {\em qualitatively}: winning with certainty (sure
semantics in the terminology of~\cite{DBLP:journals/tcs/AlfaroHK07}),
and winning with probability one (almost sure semantics in the
terminology of~\cite{DBLP:journals/tcs/AlfaroHK07}). We consider both
semantics here.

Previous papers on concurrent games are mostly concerned with two-player
zero-sum games, i.e. two players that have fully antagonistic
objectives. In this paper, we consider the {\em more general setting}
of $n$-player non zero-sum concurrent games in which each player has
its own objective. The notion of winning strategy is not sufficient to
study non zero-sum games and other solution concepts have been
proposed. One such concept is the notion of {\em admissible
  strategy}~\cite{adam2008admissibility}.

For a player with objective $\Phi$, a strategy $\sigma$ is said to be
\emph{dominated} by a strategy $\sigma'$ if $\sigma'$ does as well as
$\sigma$ with respect to $\Phi$ against all the strategies of the
other players and strictly better for some of them. A strategy
$\sigma$ is \emph{admissible} for a player if it is not dominated by
any other of his strategies. Clearly, playing a strategy which is not
admissible is sub-optimal and a rational player should only play
admissible strategies. While recent works have studied the notion of
admissibility for $n$-player non zero-sum game
graphs~\cite{Berwanger07,Faella09,BRS14,brs15,FSTTCS2016}, they are
all concerned with the special case of turn-based games and this work
is the first to consider the more general concurrent games.

Throughout the paper, we consider the running example in
\figurename~\ref{ex:runex}. This is a concurrent game played by two
players. Player~$1$'s objective is to reach $\target$, while
Player~$2$ wants to reach $s_2$.  Edges are labelled by pairs of moves
of both players which activate that transition (where $-$ means `any
move').  It is easy to see that no player can enforce its objective
with or without randomisation, so, there is no winning strategy in
this game for either player.  This is because moving from $s_0$ to
$s_1$ and from $s_1$ to $s_2$ requires the cooperation of both
players. Moreover, the transitions from $s_2$ behave as in the
classical `matching pennies' game: player $1$ must chose between $f$
and $f'$; player $2$ between $g$ and $g'$; and the target is reached
only when the choices `match'.  So, randomisation is needed to make
sure $\target$ is reached with probability one, from $s_2$. In the
paper, we will describe the dominated and admissible strategies of
this game.

\begin{figure}
 
  \begin{subfigure}{.32\textwidth}
    \centering
    \begin{tikzpicture}[yscale=.5]
      \draw [rounded corners] (1.5,1) rectangle (4,4) {} ;
      \draw [rounded corners] (2,2) rectangle (4.5,5) {} ;
      
      \node at (2.5, 1.5) {\Large\textbf{SCO}} ;
      \node at (3.5, 4.5) {\Large\textbf{LA}} ;
      \node at (3,3) {\Large\textbf{Adm.}} ;
      
      \draw [very thick] (4,2) {[rounded corners] -- (4,4)} -- (2,4) {[rounded corners] -- (2,2)} -- (4,2) ;
    \end{tikzpicture}
    \caption{Concurrent games. \label{fig:relconc}}
  \end{subfigure}
  \begin{subfigure}{.32\textwidth}
    \centering
    \begin{tikzpicture}[yscale=.5]
      \draw [rounded corners, very thick] (1.5,1) rectangle (4,4) {} ;
      \draw [rounded corners] (1.5,1) rectangle (4.5,5) {} ;
      
      \node at (3.5, 4.5) {\Large\textbf{LA}} ;
      
      \node at (2.75,2.5) {\begin{tabular}{c}\Large\textbf{Adm.}\\=\\\Large\textbf{SCO}\end{tabular}} ;
    \end{tikzpicture}
    \caption{Turn-based games.\label{fig:reltb}}
  \end{subfigure}
  \begin{subfigure}{.32\textwidth}
    \centering
    \begin{tikzpicture}[yscale=.5]
      \draw [rounded corners] (1.5,1) rectangle (4.5,5) {} ;
      \draw [rounded corners, very thick] (2,2) rectangle (4.5,5) {} ;
      
      \node at (2.5, 1.5) {\Large\textbf{SCO}} ;
      \node at (3.25,3.5) {\begin{tabular}{c}\Large\textbf{Adm.}\\=\\\Large\textbf{LA}\end{tabular}} ;
    \end{tikzpicture}
    \caption{Safety games.\label{fig:relsafety}}
  \end{subfigure}
  \caption{
  The relationships between the classes of Admissible, LA, and
    SCO strategies for three families of games. All the
    inclusions are strict.}
  \label{fig:relations}
\end{figure}
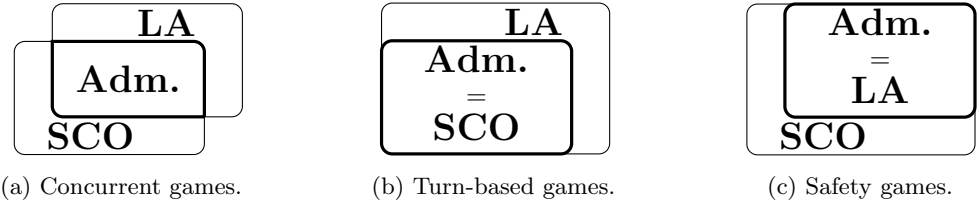

\subparagraph*{Technical contributions} First, we study the notion of
admissible strategies for both the {\em sure} and {\em almost sure}
semantics of concurrent games. We show in
Theorem~\ref{theo:classifconc} that in both semantics admissible
strategies always exist. The situation is thus similar to the
turn-based case~\cite{Berwanger07,BRS14}. Nevertheless, the techniques
used in this simpler case do not generalise easily to the concurrent
case and we need substantially more involved technical tools here. To
obtain our universal existence result, we introduce two weaker
solution concepts: \emph{locally admissible moves} and \emph{strongly
  cooperative optimal} strategies. While cooperative optimal
strategies were already introduced in~\cite{FSTTCS2016} and shown
equivalent to admissible strategies in the turn based setting, there
are strictly weaker than admissible strategies in the concurrent
setting (both for the sure and the almost sure semantics), and they
need to be combined with the notion of locally admissible moves to
fully characterise admissible strategies. In the special case of
safety objectives, we can show that admissible strategies are exactly
those that always play locally admissible moves. This situation is
depicted in \figurename~\ref{fig:relations}.

Second, we build on our characterisation of admissible strategies
based on the notions of locally admissible moves and strongly
cooperative optimal strategies to obtain algorithms to solve the {\em
  assume admissible synthesis} problem for concurrent games.  In the
assume admissible synthesis problem, we ask whether a given player has
an \emph{admissible} strategy that is \emph{winning} against \emph{all
  admissible} strategies of the other players.  So this rule relaxes
the classical synthesis rule by asking for a strategy that is winning
against the admissible strategies of the other players only and not
against all of them. This is reasonable as in a multi-player game,
each player has his own objective which is generally not the
complement of the objectives of the other players. The
assume-admissible rule makes the hypothesis that players are rational,
hence they play admissible strategies and it is sufficient to win
against those strategies.  Our algorithm is applicable to all
$\omega$-regular objectives and it is based on a reduction to a
zero-sum two-player game in the sure semantics. While this reduction
shares intuitions with the reduction that we proposed in~\cite{brs15}
to solve the same problem in the turn-based case, our reduction here
is based on games with imperfect
information~\cite{DBLP:journals/lmcs/RaskinCDH07}. In contrast, in the
turn-based case, games of perfect information are sufficient. The
correctness and completeness of our reduction are proved in
Theorem~\ref{theo:aasynth}.

\subparagraph*{Related works} Concurrent reachability games were
studied in~\cite{DBLP:journals/tcs/AlfaroHK07} and algorithms to solve
more general omega-regular objectives were given
in~\cite{DBLP:journals/tocl/ChatterjeeAH11}. The games studied there
are two-player and zero-sum only. We rely on the algorithms defined
in~\cite{DBLP:journals/tocl/ChatterjeeAH11} to compute states from
which players have almost surely winning strategies for their
objectives when all the other players play adversely. States where
a player has a (deterministic) strategy to surely win against all
other players can be computed by a reduction to more classical
turn-based game graphs~\cite{DBLP:journals/jacm/AlurHK02}. 
Nash equilibria have been studied in concurrent games~\cite{BBMU11},
but without randomised strategies.
None of
those papers consider the notion of admissibility.

We use the notion of admissibility to obtain synthesis algorithms for
systems composed of several sub-systems starting from non zero-sum
specifications. There have been a few other proposals in the
literature that are based on refinements of the notion of Nash
equilibrium (and not on admissibility), most notably: assume-guarantee
synthesis~\cite{DBLP:conf/tacas/ChatterjeeH07} and rational
synthesis~\cite{DBLP:conf/tacas/FismanKL10,DBLP:conf/eumas/KupfermanPV14}.
Those works assume the simpler setting of turn-based games and so they
do not deal with randomised strategies. In the context of infinite
games played on graphs, one well known limitation of Nash equilibria
is the existence of non-credible threats, admissibility does not
suffer from this problem.

In~\cite{DBLP:conf/fm/DammF14}, Damm and Finkbeiner use the notion of {\em dominant
  strategy} to provide a compositional semi-algorithm for the (undecidable)
  distributed synthesis problem.  So while we use the notion of admissible
  strategy, they use a notion of dominant strategy. The notion of dominant
  strategy is {\em strictly stronger}:
  every dominant strategy is admissible but an admissible strategy is not
  necessary dominant. Also, in multiplayer games with
  omega-regular objectives with complete information (as considered here),
  admissible strategies are always guaranteed to exist~\cite{Berwanger07} while
  it is not the case for dominant strategies.  



%% file: prelim.tex
\section{Preliminaries}\label{sec:preliminaries}

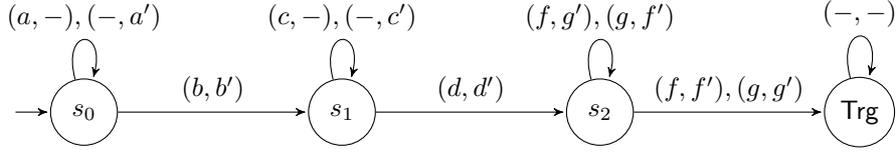
\begin{figure}
  \centering
    \begin{tikzpicture}[->,>=stealth',shorten >=1pt,auto,node distance=2.5cm, initial text={}]
      \node [initial,state] (q1)                      {$s_0$};
      \node[state]          (q2) [right=of q1]         {$s_1$};
      \node[state]          (q3) [right=of q2]         {$s_2$};    
      \node[state]          (q4) [right=of q3]         {$\target$};            
      \path (q1) edge node {$(b,b')$} (q2);
      \path (q2) edge node {$(d,d')$} (q3);          
      \path (q1) edge [loop above] node {$(a,-),(-,a')$} (q1);
      \path (q2) edge [loop above] node {$(c,-),(-,c')$} (q2);
       \path (q3) edge [loop above] node {$(f,g'),(g,f')$} (q3);
       \path (q3) edge  node {$(f,f'),(g,g')$} (q4);       
       \path (q4) edge [loop above] node {$(-,-)$} (q4);

    \end{tikzpicture}
    \caption{A concurrent game where Player 1 and 2 want
      to reach $\target$ and $s_2$ respectively.\label{ex:runex}}
\end{figure}

\subparagraph*{Concurrent games played on graphs} Let
$\players=\{1, 2,\ldots n\}$ be a set of players. A
\emph{concurrent game} played on a finite graph by the players in
$\players$ is a tuple
$\game=(\states, \Sigma,\vinit, (\choicep{p})_{p\in\players},
\transfun)$ where,
\begin{inparaenum}[(i)]
\item $\states$ is a finite set of \emph{states};
  and~$\vinit \in \states$ the \emph{initial state};
 \item $\Sigma$ is a finite set of \emph{actions}; 
 \item For all $p\in \players$,
   $\choicep{p}: \states\to 2^\Sigma\setminus\{\emptyset\}$ is a
   \emph{action assignment} that assigns, to all states $s\in \states$,
   the set of actions available to player $p$ from state $s$.
 \item
   $\transfun : \states \times \Sigma \times \ldots \times \Sigma
   \rightarrow \states$ is the \emph{transition function}.
\end{inparaenum}
We write $\Sigma(s) = \Sigma_1(s)\times \ldots \times \Sigma_n(s)$ for
all $s\in \states$.  It is often convenient to consider a player $p$
separately and see the set of all other players $P\setminus\{p\}$ as a
single player denoted $-p$.  Hence, the set of actions of $-p$ in state
$s$ is:
$\choicep{-p}(s)\defeq\prod_{q\in\players\setminus\{p\}}\choicep{q}(s)$.
We assume that~$\Sigma_i(s) \cap \Sigma_j(s) = \emptyset$ for all~$s \in S$
and~$i\neq j$.
We denote by
$\succmove{s}{a}=\{\transfun(s,a,b)\mid b\in\choicep{-p}(s)\}$ the set
of possible \emph{successors} of the state $s\in \states$ when player
$p$ performs action $a\in \choicep{p}(s)$.  A particular case of
concurrent games are the \emph{turn-based games}. A game
$\game=(\states, \Sigma,\vinit, (\choicep{p})_{p\in\players},
\transfun)$ is turn-based iff for all states $s\in \states$, there is
a unique player $p$ s.t. the successors of $s$ depend only on $p$'s
choice of action, i.e., $\succmove{s}{a}$ contains exactly one state
for all $a\in\choicep{p}(s)$.

A~\emph{history} is a finite
path~$\hist=s_1s_2\ldots s_k \in \states^*$ s.t.
\begin{inparaenum}[(i)]
\item $k\in\N$;
\item  $s_1 = \vinit$; and
\item for every~$2\leq i\leq k$, there
  exists~$(a_1,\ldots,a_n) \in \Sigma^{|P|}$ with
  $s_i = \delta(s_{i-1}, a_1,\ldots,a_n)$.
\end{inparaenum}
The \emph{length} $|\hist|$ of a
history $\hist=s_1s_2\ldots s_k$ is its number of states $k$; for
every $1\leq i\leq k$, we denote by $\hist_i$ the state $s_i$ and by
$\hist_{\leq i}$ the history $s_1s_2\ldots s_i$.  We denote by
$\last{\hist}$ the last state of $\hist$, that is,
$\last{\hist}=\hist_{|\hist|}$.
A~\emph{run} is defined similarly as a history except that its length
is infinite. For a run $\rho=s_1s_2\ldots\in \states^\omega$ and
$i\in\N$, we also write $\rho_{\leq i}=s_1s_2\ldots s_i$ and
$\rho_i=s_i$.  Let~$\prefixoutc(\game)$ (resp. $\Runs(\game)$) denote
the set of histories (resp. runs) of $\game$.  The game is played from
the initial state $\vinit$ for an infinite number of rounds, producing
a run.  At each round~$i\geq 0$, with current state $s_i$, all
players~$p$ select simultaneously a action $a^i_p\in\choicep{p}(s_i)$,
and the state $\delta(s_i,a^i_1,\ldots,a^i_n)$ is appended to the
current history. The selection of the action by a player is done
according to strategies defined below.

\subparagraph*{Randomised moves and strategies} Given a finite set
$A$, a probability distribution on $A$, is a function
$\alpha : A\to [0,1]$ such that $\sum_{a\in A}\alpha(a)=1$; and we let
$\supp(\alpha)=\{a\mid \alpha(a)>0\}$ be the support of $\alpha$.  We
denote by $\alpha(B)=\sum_{a\in B}\alpha(a)$ the probability of a
given set $B$ according to $\alpha$.  The set of probability
distributions on $A$ is denoted by $\distr{A}$.  
A \emph{randomised move} of player~$p$ in state~$s$ 
is a probability distribution on
$\choicep{p}(s)$, that is, an element of $\distr{\choicep{p}(s)}$.  A
randomised move that assigns probability $1$ to an action and $0$ to
the others is called a \emph{Dirac move}. We will henceforth denote
randomised moves as sums of actions weighted by their respective
probabilities. For instance $0.5f+0.5g$ denotes the randomised move
that assigns probability $0.5$ to $f$ and $g$ (and $0$ to all other
actions). In particular, we denote by $b$ a Dirac move that assigns
probability $1$ to action $b$.

Given a state $s$ and a tuple
$\vec \beta=(\beta_p)_{p\in\players}\in \prod_{p\in\players}
\distr{\choicep{p}(s)}$ of randomised moves from~$s$, one per player,
we let $\transfunrand(s,\vec \beta)\in \distr{\states}$ be the
probability distribution on states s.t. for all $s'\in \states$:
$\transfunrand(s,\vec \beta) (s')=\sum_{\vec a \mid \transfun(s,\vec
  a)=s'} \vec \beta(\vec a)$, where
$\vec \beta(a_1,\ldots, a_n)=\prod_{i=1}^n\beta_i(a_i)$. Intuitively,
$\transfunrand(s,\vec \beta)(s')$ is the probability to reach $s'$
from $s$ when the players play according to $\vec \beta$.

A \emph{strategy} for player~$p$ is a function $\strat$ from histories
to randomised moves (of player $p$) such that, for all
$\hist \in \prefixoutc(\game)$:
$\strat(\hist)\in \distr{\choicep{p}(\last{\hist})}$.  A strategy is
called Dirac at history $\hist$, if $\strat(\hist)$ is a Dirac move;
it is called Dirac if it is Dirac in all histories.  We denote
by~$\rectset_p(\game)$ the set of player-$p$ strategies in the game,
and by $\rectset^{{\it det}}_p(\game)$ the set of player-$p$
strategies that only use Dirac moves (those strategies are also called
\emph{deterministic}); we might omit~$\game$ if it is clear from
context.  A~\emph{strategy profile~$\profile$} for a
subset~$A\subseteq P$ of players is a tuple $(\profile_p)_{p \in A}$
with~$\profile_p \in \rectset_p$ for all $p\in A$.  When the set of
players~$A$ is omitted, we assume $A=P$. Let
$\profile=(\profile_p)_{p \in P}$ be a strategy profile. Then, for all
players $p$, we let $\profile_{-p}$ denote the restriction of
$\profile$ to $P\setminus\{p\}$ (hence, $\profile_{-p}$ can be
regarded as a strategy of player $-p$ that returns, for all histories
$\hist$, a randomised move from
$\prod_{p\in\players\setminus\{p\}} \distr{\choicep{p}(s)}\subseteq
\distr{\choicep{-p}(\last{\hist})}$).  We sometimes denote $\profile$
by the pair $(\profile_p,\profile_{-p})$.  Given a history $\hist$, we
let $(\profile_p)_{p \in A}(\hist)=(\profile_p(\hist))_{p\in A}$.

Let
$\hist$ be a history and let $\rho$ be a history or a run. Then, we
write $\hist \prefix \rho$ iff~$\hist$ is a prefix of~$\rho$, i.e.,
$\rho_{\leq |\hist|}=\hist$.  
Consider two strategies $\strat$ and $\strat'$ for player $p$, and a
history~$\hist$. We denote by $\switchstrat{\strat}{\hist}{\strat'}$ the
strategy that follows strategy~$\strat$ and \emph{shifts} to~$\strat'$ as
soon as~$\hist$ has been played.  Formally,
$\switchstrat{\strat}{\hist}{\strat'}$ is the strategy s.t., for all
histories $\hist'$: 
$\switchstrat{\strat}{\hist}{\strat'}(\hist') = \strat'(\hist')$ if~$\hist\prefix \hist'$;
and $\switchstrat{\strat}{\hist}{\strat'}(\hist')=\strat(\hist')$ otherwise.

\subparagraph*{Probability measure and outcome of a profile} Given a
history $\hist$, we let $\cyl{\hist}=\{\rho\mid\hist\prefix\rho\}$ be
the \emph{cylinder} of $\hist$. To each strategy profile $\profile$,
we associate a probability measure $\Prob_{\profile}$ on certain sets
of runs. First, for a history $\hist$, we define
$\Prob_\profile(\cyl{\hist})$ inductively on the length of $\hist$:
$\Prob_\profile(\cyl{\vinit})=1$, and
$\Prob_{\profile}(\cyl{\hist' s'})=
\Prob_{\profile}(\cyl{\hist'})\cdot
\transfunrand(\last{\hist'},\profile(\hist'))(s')$when $|\hist|>1$
and $\hist=\hist's'$.  Based on this definition, we can extend the
definition of $\Prob_\profile$ to any Borel set of runs on
cylinders. In particular, the function $\Prob_\profile$ is
well-defined for all $\omega$-regular sets of runs, that we will
consider in this paper~\cite{vardi1985automatic}.  We extend the
$\prefixoutc$ notation and let $\supphist{\profile}$ be the set of
histories $\hist$ such that $\Prob_{\profile}(\cyl{\hist})>0$.  Given
a profile $\profile$ we denote by $\outcome{\profile}$ the set of runs
$\rho$ s.t. all prefixes $\hist$ of $\rho$ belong to
$\supphist{\profile}$. In particular,
$\Prob_{\profile}(\outcome{\profile})=1$.  Note that when $\profile$
is composed of Dirac strategies then $\outcome{\profile}$ is a
singleton.  The outcome (set of histories) of a strategy
$\strat\in\rectset_p$, denoted by $\outcome{\strat}$
($\prefixoutc(\strat)$), is the union of outcomes (set of histories,
respectively) of profile $\profile$ s.t. $\profile_p=\strat$.

\subparagraph*{Winning conditions} To determine the gain of all
players in the game $\game$, we define \emph{winning conditions} that
can be interpreted with two kinds of semantics denoted by the symbols
$\sure$ for the \emph{sure semantics} or and $\asure$ for the
\emph{almost sure semantics}.  A winning condition $\Phi$ is a subset
of $\Runs(\game)$ called \emph{winning runs}.  From now on, we assume
that concurrent games are equipped with a function $\Phi$, called the
winning condition, and mapping all players $p\in\players$ to a winning
condition $\Phi(p)$. 
A profile $\profile$ is $\asure$-winning for
$\Phi(p)$ if $\Prob_{\profile}( \Phi)=1$ which we write
$\game,\profile\rmodels{\Phi(p)}$.  A profile $\profile$ is
$\sure$-winning for $\Phi(p)$ if
$\outcome{\game,\profile}\subseteq\Phi(p)$
which
we write $\game,\profile\smodels{\Phi(p)}$.  Note that when $\profile$
is Dirac, the two semantics coincide:
$\game,\profile\smodels{\Phi(p)}$ iff
$\game,\profile\rmodels{\Phi(p)}$.  The profile $\profile$ is
$\asure$-winning from $\hist$ if $\hist\in\prefixoutc(\game,\profile)$
and
$\Prob_{\profile}( \Phi(p) \mid \cyl{\hist})=\Prob_{\profile}( \Phi(p)
\cap\cyl{\hist})/\Prob_{\profile}(\cyl{\hist}) =1$ which we denote
$\game,\profile\rmodelsh{\hist}{\Phi(p)}$.
The profile $\profile$ is winning for the sure semantics from $\hist$ if
$\{ \rho \in \outc(\game,\profile) \mid h \prefix \rho\} \subseteq \Phi(p)$,
which we denote $\game,\profile\smodelsh{\hist}{\Phi(p)}$.  We often
omit $\game$ in notations when clear from the context.  Most of our
definitions and results hold for both semantics and we often state
them using the symbol $\star\in\{\sure,\asure\}$ as in the following
definition.  Given a semantics $\star\in\{\sure,\asure\}$, a strategy
$\strat$ for player $p$ (from a history $\hist$) is called
$\star$-winning for player $p$ if for every $\tau\in\rectset_{-p}$,
the profile $(\strat,\tau)$ is $\star$-winning for player $p$ (from
$\hist$).  Note that a strategy $\strat$ for player $p$ is
$\sure$-winning iff $\outcome{\strat}\subseteq \Phi(p)$.  We often
describe winning conditions using standard linear temporal operators
$\Box$ and~$\Diamond$; \textit{e.g.} $\Box\Diamond S$ means the set of
runs that visit infinitely often~$S$. See~\cite{BK2008} for a formal
definition.

A winning condition~$\Phi(p)$ is \emph{prefix-independent} if for
all~$s_1s_2\ldots \in \Phi(p)$, and all~$i\geq 1$:
$s_is_{i+1}\ldots \in \Phi(p)$.  When $\Phi(p)$ contains all runs that
do not visit some designated set $\mathsf{Bad}_p\subseteq \states$ of
states, we say that $\Phi(p)$ is a \emph{safety condition}. A safety
game is a game whose winning condition $\Phi$ is such that $\Phi(p)$
is a safety condition for all players $p$. Without loss of generality,
we assume that safety games are so-called \emph{simple} safety games:
a safety game
$(\states, \Sigma,\vinit, (\choicep{p})_{p\in\players}, \transfun)$ is
\emph{simple} iff for all players $p$, for all $s\in \states$:
$s\in\mathsf{Bad}_p$ implies that no $s'\not\in\mathsf{Bad}_p$ is
reachable from $s$. That is, once the safety condition is violated,
then it remains violated forever at all future histories.

We note the following property of winning strategies.
\begin{lemma}\label{lem:winsforeveryhist}
  Given $\star\in\{\sure,\asure\}$, $\profile\starmodels{\Phi(p)}$ iff
  $\profile\starmodelsh{\hist}{\Phi(p)}$ for every
  $\hist\in \supphist{\profile}$.
\end{lemma}
\begin{proof}
  The only non-trivial implication is: if $\profile\rmodels{\Phi(p)}$
  then $\profile\rmodelsh{\hist}{\Phi(p)}$ for every
  $\hist\in \supphist{\profile}$ which we show by contraposition.
  Assume there exists $\hist\in \supphist{\profile}$ such that
  $\profile\not \rmodelsh{\hist}{\Phi(p)}$. This means that
  $\Prob_{\profile}(\cyl{\hist})>0$ and
  $\Prob_{\profile}( \Phi(p)
  \cap\cyl{\hist})/\Prob_{\profile}(\cyl{\hist})<1$.  It follows that
  $\profile\not \rmodelsh{\hist}{\Phi(p)}$.
\end{proof}

\begin{example}\label{example:running}
  Let us consider three player-$1$ strategies in Fig.~\ref{ex:runex}.
  \begin{inparaenum}[(i)]
  \item $\strat_1$ is any strategy that plays $a$ in $s_0$;
  \item $\strat_2$ is any strategy that plays $b$ in $s_0$, $d$ in
    $s_1$ and $f$ in $s_2$; and
  \item $\strat_3$ is any strategy that plays $b$ in $s_0$, $d$ in
    $s_1$, and $0.5f+0.5g$ in $s_2$.
  \end{inparaenum}
  Clearly, $\strat_1$ never allows one to reach $\target$ while some
  runs respecting $\strat_2$ and $\strat_3$ do (remember that there is
  no $\star$-winning strategy in this game).  We will see later that
  the best choice of player $1$ (among $\strat_2$, $\strat_3$) depends
  on the semantics we consider.  In the almost-sure semantics,
  $\strat_3$ is `better' for player~$1$, because $\strat_3$ is an
  $\asure$-winning strategy from all histories ending in $s_2$, while
  $\strat_2$ is not.  On the other hand, in the sure semantics,
  playing $\strat_2$ is 'better' for player $1$ than
  $\strat_3$. Indeed, for all player-$2$ strategies $\tau$, either
  $\outcome{\strat_3,\tau}$ contains only runs that do not reach
  $s_2$, or $\outcome{\strat_3,\tau}$ contains at least a run that
  reaches $s_2$, but, in this case, it also contains a run of the form
  $\hist s_2^\omega$ (because, intuitively, player $1$ plays
  \emph{both} $f$ and $g$ from $s_2$). So, $\strat_3$ is not winning
  against any $\tau$, while $\strat_2$ wins at least against a player
  $2$ strategy that plays $b'$ in $s_0$, $d'$ in $s_1$ and $f'$ in
  $s_2$. We formalise these intuitions in the next section.
\end{example}


%% file: adm.tex
\section{Admissibility} \label{sec:admissibility} In this section, we
define the central notion of the paper: admissibility
\cite{Berwanger07,BrenguierRS17}. Intuitively, a strategy is
admissible when it plays `as well as possible'. Hence the definition
of admissible strategies is based on a notion of domination between
strategies: a strategy $\strat'$ dominates another strategy $\strat$
when $\strat'$ wins every time $\strat$ does. Obviously, players have
no interest in playing dominated strategies, hence admissible
strategies are those that are not dominated. Apart from these
(classical) definitions, we characterise admissible strategies as
those that satisfy two weaker notions: they must be both
\emph{strongly cooperative optimal} and play only
\emph{locally-admissible moves}. Finally, we discuss important
characteristics of admissible strategies that will enable us to
perform assume-admissible synthesis (see Section~\ref{sec:assume-admiss-synth}).


In this section, we fix a game $\game$, a player $p$, and, following
our previous conventions, we denote by $\rectset_{-p}$ the set
$\{\profile_{-p}\mid \profile\in\rectset\}$.

\subparagraph*{Admissible strategies}
We first recall the classical notion of \emph{admissible strategy}
\cite{Berwanger07,adam2008admissibility}. Given two
strategies $\strat,\strat'\in\rectset_p$, 
we say
that $\strat$ is $\star$-\emph{weakly dominated} by $\strat'$, denoted
$\strat \bothdomstrateq{\star} \strat'$, if for all
$\tau\in\rectset_{-p}$: $(\strat, \tau) \starmodels{\Phi(p)}$
\emph{implies} $(\strat', \tau) \starmodels{\Phi(p)}$.  This indeed
captures the idea that $\strat'$ is \emph{not worse} that $\strat$,
because it wins (for $p$) every time $\strat$ does.  Note that
$\bothdomstrateq{\star}$ is not anti-symmetric, hence we note
$\strat\botheqadmstrat{\star} \strat'$ when $\strat$ and $\strat'$ are
equivalent, i.e.  $\strat\bothdomstrateq{\star} \strat'$ and
$\strat'\bothdomstrateq{\star} \strat$.  In other words
$\strat\botheqadmstrat{\star}\strat'$ iff for every
$\tau\in \rectset_{-p}$,
$(\strat,\tau) \starmodels{ \Phi(p)}\Leftrightarrow(\strat',\tau)
\starmodels{\Phi(p)}$.  When $\strat \bothdomstrateq{\star} \strat'$
but $\strat' \not \bothdomstrateq{\star} \strat$ we say that $\strat$
is \emph{$\star$-dominated} by $\strat'$, and we write
$\strat \bothdomstrat{\star} \strat'$. Observe that
$\strat \bothdomstrat{\star} \strat'$ holds if and only if
$\strat \bothdomstrateq{\star} \strat'$ and there exists at least one
$ \tau\in\rectset_{-p}$, such that
$(\strat, \tau) \not \starmodels{\Phi(p)}$ and
$(\strat', \tau) \starmodels{\Phi(p)}$.  That is, $\strat'$ is now
\emph{strictly better} than $\strat$.  Then, \emph{a strategy $\strat$
  is \textbf{$\star$-admissible} iff there is no strategy $\strat'$
  s.t. $\strat \bothdomstrat{\star} \strat'$}, i.e., $\strat$ is
$\star$-admissible iff it is not $\star$-dominated.

\begin{example} \label{example:admissible} Let us continue our running
  example, by formalising the intuitions we have sketched in
  Example~\ref{example:running}. Since $\strat_1$ does not allow to
  reach the target, while some runs respecting $\strat_2$ and
  $\strat_3$ do, we have: $\strat_1\bothdomstrat{\star} \strat_2$ and
  $\strat_1\bothdomstrat{\star} \strat_3$. Moreover,
  $\strat_2\bothdomstrat{\asure} \strat_3$ because $\strat_3$ is
  $\asure$-winning from any history that ends in $s_2$ while
  $\strat_2$ is not because it does not $\asure$-win against a
  player~$2$ strategy that would always play $g'$ in $s_2$ (and both
  strategies behave the same way in $s_0$ and $s_1$).
  On the other hand, $\strat_3\bothdomstrat{\sure} \strat_2$. In fact,
  $\outcome{\strat_2}\subsetneq\outcome{\strat_3}$, so, for all
  $\tau \in \rectset_2$, whenever $(\strat_3,\tau) \smodels \Phi(1)$,
  all runs of $\outcome{\strat_3}$ reach $\target$, and we also have
  $(\strat_2, \tau) \smodels \Phi(1)$.  Moreover, for the strategy
  $\tau\in\rectset_{2}$ that plays $f'$ at~$s_2$, profile
  $(\strat_2,\tau)$ is
  $\sure$-winning but not
  $(\strat_3,\tau)$. We will see later that $\strat_3$ is
  $\asure$-admissible and $\strat_2$ is $\sure$-admissible.
\end{example}

\subparagraph*{Values of histories} Before we discuss strongly
cooperative optimal and locally admissible strategies, we associate
\emph{values} to histories. Let $\hist$ be a history, and $\strat$ be
a strategy of player~$p$. Then, the \emph{value of $\hist$ w.r.t.
  $\strat$ for semantics} $\star\in\{\sure,\asure\}$ is defined as follows.
$\valhstrat{\star}{\hist}{\strat}= 1$ if $\strat$ is $\star$-winning from $\hist$;
$\valhstrat{\star}{\hist}{\strat}= 0$ if there are $\tau\in\rectset_{-p}$ and
$\tau'\in\rectset_{-p}$ s.t. 
$(\strat,\tau)\starmodelsh{\hist}{\Phi(p)}$, and $(\strat,\tau')\not\starmodelsh{\hist}{\Phi(p)}$;
and~$-1$ otherwise.


Value $\valhstrat{\star}{\hist}{\strat}=1$ corresponds to the case
where $\strat$ is $\star$-winning for player $p$ from $\hist$ (thus,
against all possible strategies in $\rectset_{-p}$). When
$\valhstrat{\star}{\hist}{\strat}=0$, $\strat$ is \emph{not}
$\star$-winning from $\hist$ (because of $\tau'$ in the definition),
but the other players can still help $p$ to reach his objective (by
playing some $\tau$ s.t. $(\strat,\tau)\starmodelsh{\hist}{\Phi(p)}$,
which exists by definition). Last,
$\valhstrat{\star}{\hist}{\strat}=-1$ when there is no hope for $p$ to
$\star$-win, even with the collaboration of the other players. In this
case, there is no $\tau$
s.t. $(\strat,\tau)\starmodelsh{\hist}{\Phi(p)}$. Hence, having
$\valhstrat{\star}{\hist}{\strat}=-1$ is stronger than saying that
$\strat$ is not winning---when $\strat$ is not winning, we could have
$\valhstrat{\star}{\hist}{\strat}=0$ as well.

We define the value of a history $\hist$
for player $p$ as the best value he can achieve with his different
strategies:
  $\valh{\star}{\hist}{p} =\max_{\strat\in
                       \rectset_{p}}\valhstrat{\star}{\hist}{\strat}$.
Last, for
$v\in\{-1, 0,1\}$, let $\bothvalpred{\star}{p}{v}$ be the set of
histories $\hist$ s.t. $\valh{\star}{\hist}{p}=v$.

\subparagraph*{Strongly cooperative optimal strategies} We are now
ready to define \emph{strongly cooperative optimal} (SCO) strategies.
Recall that, in the classical setting of turn-based games, admissible
strategies are exactly the SCO strategies \cite{BrenguierRS17}. We
will see that this condition is still necessary but not sufficient in
the concurrent setting.

A strategy $\strat$ of Player~$p$ is $\star$-SCO at $\hist$ iff
$\valhstrat{\star}{\hist}{\strat}=\valh{\star}{\hist}{p}$; and
$\strat$ is $\star$-SCO iff it is $\star$-SCO at all
$\hist\in\prefixoutc(\strat)$. Intuitively, when $\strat$ is a
$\star$-SCO strategy of Player~$p$, the following should hold:
\begin{inparaenum}[(i)]
\item if $p$ has a $\star$-winning strategy from $\hist$ (i.e.
  $\valh{\star}{\hist}{p}=1$), then, $\sigma$ should be
  $\star$-winning (i.e. $\valhstrat{\star}{\hist}{\strat}=1$);
  and
\item otherwise if $p$ has no $\star$-winning strategy from $h$ but
  still has the opportunity to $\star$-win with the help of other
  players (hence $\valh{\star}{\hist}{p}=0$), then, $\sigma$
  should enable the other players to help $p$
  fulfil his objective (i.e. $\valhstrat{\star}{\hist}{\strat}=0$).
\end{inparaenum}
Observe that when $\valh{\star}{\hist}{p}=-1$, no continuation of
$\hist$ is $\star$-winning for $p$, so
$\valhstrat{\star}{\hist}{\strat}=-1$ for all strategies $\sigma$.

\begin{example}\label{example:sco}
  Consider again the example in \figurename~\ref{ex:runex}. For the
  almost-sure semantics, we have
  $\bothvalpred{\asure}{p}{1}=\big\{\hist\mid
  \last{\hist}\in\{s_2,\target\}\big\}$, and
  $\bothvalpred{\asure}{p}{0}=\big\{\hist\mid
  \last{\hist}\in\{s_0,s_1\}\big\}$. For the sure semantics, we have:
  $\bothvalpred{\sure}{1}{1}=\{\hist\mid\last{\hist}=\target\}$, and
  $\bothvalpred{\sure}{1}{0}=\{\hist\mid\last{\hist}\neq
  \target\}$. Let us Consider again the three strategies $\strat_1$,
  $\strat_2$ and $\strat_3$ from Example~\ref{example:running}. We see
  that $\strat_2$ is $\sure$-SCO but it is not $\asure$-SCO because,
  for all profiles $\hist$ ending in $s_2$:
  $\valh{\asure}{\hist}{\strat_2}=0$ while
  $\hist\in \bothvalpred{\asure}{1}{1}$. On the other hand, $\strat_3$
  is $\asure$-SCO; but it is not $\sure$-SCO. Indeed, one can check
  that, for all strategies $\tau\in\rectset_2$: if
  $\outcome{\strat_3,\tau}$ contains a run reaching $\target$, then it
  also contains a run that cycles in $s_2$. So, for all such
  strategies $\tau$, $\outcome{\strat_3,\tau}\not\smodels \win{1}$,
  hence $\valh{\sure}{\hist}{\strat_3}=-1$ for all histories that end
  in $s_2$; while~$\valh{\sure}{\hist}{p}=0$ since
  $\valh{\sure}{\hist}{\sigma'}=0$ for all Dirac strategies~$\sigma'$.

  Next, let us build a strategy $\strat_3'$ that is $\asure$-dominated
  by $\strat_3$ (hence, not $\asure$-admissible), but $\asure$-SCO. We let
  $\strat_3'$ play as $\strat_3$ except that $\strat_3'$ plays $c$ the
  first time $s_1$ is visited (hence ensuring that the self-loop on
  $s_1$ will be taken after the first visit to $s_1$). Now, $\strat_3$
  is $\asure$-dominated by $\strat_3'$, because
  \begin{inparaenum}[(i)]
  \item $\strat_3$ $\asure$-wins every time $\strat_3'$ does; but
  \item $\strat_3'$ does not $\asure$-win against the player $2$
    strategy $\tau$ that plays $d'$ only when $s_1$ is visited for the
    first time, while $\strat_3$ $\asure$-wins against $\tau$.
  \end{inparaenum}
  However, $\strat_3'$ is SCO because playing $c$ keeps the value of
  the history equal to $0=\valh{\asure}{\hist}{1}$ (intuitively, playing
  $c$ once does not prevent the other players from helping in the
  future).  As similar example can be built in the $\sure$
  semantics. Thus, \textbf{there are $\star$-SCO strategies which are
    not admissible}, so, being $\star$-SCO is not a sufficient
  criterion for admissibility.

\end{example}

\subparagraph*{Locally admissible moves and strategies}
Let us now discuss another criterion for admissibility, which is more
\emph{local} in the sense that it is based on a domination between
\emph{moves} available to each player after a given history. Let
$\hist$ be a history, and let $\alpha$ and $\alpha'$ be two randomised
moves in $\distr{\choicep{p}}$. We say that $\alpha$ is
$\star$-\emph{weakly dominated} at $\hist$ by $\alpha'$ (denoted
$\alpha\bothdommoveeq{\star}{\hist} \alpha'$) iff for
all~$\strat \in \rectset_p$ such that $h \in \prefixoutcome{\strat}$
and $\strat(h) = \alpha$, there exists~$\strat' \in \rectset_p$ s.t.
$\strat'(h) = \alpha'$ and $\strat \bothdomstrateq{\star} \strat'$.
Observe that the relation $\bothdommoveeq{\star}{\hist}$ is not
anti-symmetric.  We let $\botheqadm{\star}{\hist}$ be the equivalence
relation s.t. $\alpha\botheqadm{\star}{\hist}\beta$ iff
$\alpha\bothdommoveeq{\star}{\hist}\beta$ and
$\beta\bothdommoveeq{\star}{\hist}\alpha$.  When
$\alpha \bothdommoveeq{\star}{\hist} \alpha'$ but
$\alpha' \not \bothdommoveeq{\star}{\hist} \alpha$ we say that
$\alpha$ is $\star$-\emph{dominated} at $\hist$ by $\alpha'$ and
denote this by $\alpha \bothdommove{\star}{\hist} \alpha'$.  When a
randomised move $\alpha$ is not $\star$-dominated at $\hist$, we say
that $\alpha$ is \emph{$\star$-admissible} at $\hist$.  This allows us
to define a more local notion of dominated strategy: \emph{a strategy
  $\strat$ of player $p$ is \textbf{$\star$-locally-admissible} (LA)
  if $\strat(\hist)$ is a $\star$-admissible move at $h$, for all
  histories $\hist$}.

\begin{example}\label{example:locally-admissible}
   Consider the Dirac move $f$ and the non-Dirac move $0.5f+0.5g$
  played from $s_2$ in the example in \figurename~\ref{ex:runex}. One
  can check that $0.5 f +0.5 g \bothdommove{\sure}{s_2}f$. Indeed,
  consider a strategy $\strat$ s.t. $\strat(\hist)=0.5 f + 0.5 g$ for
  some $\hist$ with $\last{\hist}=s_2$. Then, playing $\strat(\hist)$
  from $\hist$ will never allow Player 1 to reach $\target$ \emph{surely}
  at the next step, whatever Player 2 plays; while playing, for instance, $f$ (Dirac
  move) ensures player 1 to reach $\target$ surely at the
  next step, against a Player-2 strategy that plays $f'$.
  Thus, $\sigma_2$ is $\sure$-LA but $\sigma_3$ is not.
  
  On the other hand, after every randomised move played in state
  $s_2$, the updated state is $s_2$ or $s_3$ from which
  $\asure$-winning strategies exist, thus
  $f\botheqadm{\asure}{\hist}g \botheqadm{\asure}{\hist} \lambda f
  +(1-\lambda) g$ for all $\lambda \in [0,1]$ and all histories
  $\hist$ s.t. $\last{\hist}=s_2$ (so, in particular,
  $\lambda f +(1-\lambda) g\bothdommoveeq{\asure}{\hist} f$ and
  $\lambda f +(1-\lambda) g\bothdommoveeq{\asure}{\hist} g$).  It
  follows that both~$\sigma_2$ and~$\sigma_3$ are~$\asure$-LA.
  However, in the long run, player 1 needs to play
  $\lambda f +(1-\lambda) g$, with $\lambda\in(0,1)$, infinitely often
  in order to $\asure$-win. In fact, $\sigma_3$ is $\asure$-winning
  from~$s_2$ while $\strat_2$ is not. Thus, \textbf{there are
    $\star$-LA strategies which are not admissible}, so being
  $\star$-LA is not a sufficient criterion for $\star$-admissibility.


\end{example}

We close this section by several lemmata that allow us to better
characterise the notion of LA strategies. First, we observe that,
while randomisation might be necessary for winning in certain
concurrent games
(for example, in
  \figurename~\ref{ex:runex}, no Dirac move allows player 1 to reach
  $\target$ \emph{surely} from $s_2$, while playing repeatedly $f$ and $g$ with equal
  probability ensures to reach $\target$ with probability 1)
randomisation is useless when a player wants to play only locally
admissible moves. This is shown by the next Lemma
(point~(\ref{item:4})), saying that, \emph{if a randomised move
  $\alpha$ plays some action $a$ with some positive probability, then
  $\alpha$ is dominated by the Dirac move $a$}. However, this does
not immediately allow us to characterise admissible moves: some Dirac
moves could be dominated (hence non-admissible), and some non-Dirac
moves could be admissible too. Points~(\ref{item:3})
and~(\ref{it3:characadmmoves}) elucidate this: \emph{among Dirac
  moves, the non-dominated ones are admissible}, and \emph{a non-Dirac
  move is admissible iff all the Dirac moves that occur in its support
  are admissible and equivalent to each other}.

\begin{lemma}\label{lem:Diracdominates}\label{lem:characadmmoves}
  For all histories $\hist$ and all randomised moves $\alpha$:
  \begin{enumerate}[(i)]
  \item \label{item:4} For all $a\in \supp(\alpha)$:
    $\alpha\bothdommoveeq{\star}{\hist}a$;
  \item \label{item:3} Dirac moves that are not $\star$-dominated at
    $\hist$ by another Dirac move are admissible;
  \item\label{it3:characadmmoves} A move $\alpha$ is $\star$-LA at
    $\hist$ iff, for all $a\in\supp(\alpha)$:
    \begin{inparaenum}[(1)]
    \item 
      $a$ is $\star$-LA at $\hist$; and
    \item $a\botheqadm{\star}{\hist}b$ for all $b \in\supp(\alpha)$.
    \end{inparaenum}
  \end{enumerate}
\end{lemma}
\begin{proof}
  \noindent\textbf{Proof of (\ref{item:4}):}
  Take a strategy $\strat \in \rectset_p$ such that
  $h \in \prefixoutcome{\strat}$ and $\strat(h) = \alpha$.  Define
  $\strat'$ the strategy that plays as $\strat$ except in $\hist$
  where it plays $a$ instead of $\alpha$.  One show that
  $\strat \bothdomstrateq{\star} \strat'$.  Consider
  any~$\tau \in \rectset_{-p}$ such
  that~$(\sigma,\tau) \rmodelsh{\hist}{\Phi(p)}$.  This means that
  $\Prob_{(\sigma,\tau)}(\Phi(p) \mid \cyl{\hist})=1$, so, in
  particular
  $\Prob_{(\sigma,\tau)}( \Phi(p) \mid
  \cyl{\hist\delta(s,a,\tau(h))})=1$, for all~$a \in
  \supp(\alpha)$. Since~$\sigma$ and~$\sigma'$ are identical in all
  other histories, including those extending~$h\delta(s,a,\tau(h))$,
  we deduce that~$(\sigma',\tau) \rmodelsh{\hist}{\Phi(p)}$.  The
  proof for the sure semantics is similar.

  \noindent\textbf{Proof of (\ref{item:3}):} 
  If a Dirac move $a$ is $\star$-dominated at $\hist$ by a move
  $\alpha'$ then by (\ref{item:4}) it is
  $\star$-dominated at $\hist$ by a Dirac move $a'\in \supp(\alpha')$.
  This shows (\ref{item:3}) by contraposition.

  \noindent\textbf{Proof of (\ref{it3:characadmmoves}):}
  Assume $\alpha$ is a $\star$-LA move at $\hist$. By 
  (\ref{item:4}), for every $a\in \supp(\alpha)$,
  $\alpha\bothdommoveeq{\star}{\hist} a$ and
  $\alpha\not \bothdommove{\star}{\hist} a$ (because $\alpha$ is a
  $\star$-LA move at $\hist$) so $\alpha \botheqadm{\star}{\hist}
  a$. Hence all the elements of $\supp(\alpha)$ are equivalent and
  $\star$-LA at $\hist$.

  Assume now that $\forall a \in\supp(\alpha)$, $a$ is $\star$-LA at
  $\hist$ and $\forall b \in\supp(\alpha)$,
  $a\botheqadm{\star}{\hist}b$.  We take $a \in\supp(\alpha)$. By
  Lemma \ref{lem:Diracdominates} it holds that
  $\alpha \bothdommoveeq{\star}{\strat} a$; so it remains to show that
  $a \bothdommoveeq{\star}{\hist} \alpha$. Let $\strat$ be such that
  $\strat(\hist)=a$. We construct a strategy $\strat'$ such that
  $\strat'(\hist)=\alpha$ and $\strat\bothdomstrateq{\star}\strat'$.  For
  every $b \botheqadm{\star}{\hist} a$, there exists a strategy
  $\strat_b$ such that $\strat_b(\hist)=b$ and
  $\strat\bothdomstrateq{\star}\strat_b$.  We construct $\strat'$ as
  follows. We start with $\strat'=\strat$, we set
  $\strat'(\hist)=\alpha$, for every $b\in\supp(\alpha)$ and for every
  $s'\in\succmove{s}{b}$, we do
  $\strat'\leftarrow \switchstrat{\strat'}{\hist s'}{\strat_b}$.  This
  shows that $a \bothdommoveeq{\star}{\hist} \alpha$.  We conclude
  that $\alpha$ is $\star$-LA at $\hist$.
\end{proof}

\begin{example}
  As we have seen in Example~\ref{example:locally-admissible},
  $0.5 f +0.5 g \bothdommove{\sure}{s_2}f$. Note that a strategy
  $\strat'$ s.t. $\strat'(\hist)= 0.5 f +0.5 g$ for all $\hist$ with
  $\last{\hist}=s_2$ has value $\valhstrat{\sure}{\hist}{\strat'}=-1$,
  while $\valh{\sure}{\hist}{1}=0$.
\end{example}
This example seems to suggest that the local dominance of two moves
coincide with the natural order on the values of histories that are
obtained when playing those moves (in other words
$x \bothdommove{\star}{\hist} y$ would hold iff the value of the
history obtained by playing $x$ is smaller than or equal to the value
obtained by playing $y$). This is not true for histories of value $0$:
we have seen that $a$ and $b$ are
$\bothdommoveeq{\star}{\hist}$-incomparable, yet playing $a$ or $b$
from $s_0$ yields a history with value $0$ in all cases (even when
$s_1$ is reached). The next Lemma gives a precise characterisation of
the dominance relation between Dirac moves in terms of values:

\begin{lemma}\label{lem:dominemove}
  For all players $p$, histories $\hist$ with $\last{\hist}=s$ and
  Dirac moves $a,b\in\choicep{p}(s)$:
  $a\bothdommoveeq{\star}{\hist} b$ if, and only if the following
  conditions hold for every $c\in \choicep{-p}(s)$ where we write
  $s_{(a,c)}=\transfun(s,(a,c))$ and $s_{(b,c)}=\transfun(s,(b,c))$:
  \begin{enumerate}[(i)]
  \item\label{item:1}
    $\valh{\star}{\hist s_{(a,c)}}{p}\leq \valh{\star}{\hist
      s_{(b,c)}}{p}$;
  \item\label{item:2} if
    $\valh{\star}{\hist s_{(a,c)}}{p}= \valh{\star}{\hist
      s_{(b,c)}}{p}=0$ then
    $s_{(a,c)}=s_{(b,c)}$. 
\end{enumerate}
\end{lemma}
\begin{proof}
  \noindent\textbf{Direction $a\bothdommoveeq{\star}{\hist}b\Rightarrow$(\ref{item:1})$\wedge$(\ref{item:2}):}
  Assume that $a\bothdommoveeq{\star}{\hist}b$.  Let
  $\strat\in\rectset_{p}$ such that $\hist\in\prefixoutc(\strat)$ and
  $\strat(\hist)=a$ and strategy $\strat'\in\rectset_{p}$ such that
  $\strat'(\hist)=b$ and $\strat\bothdomstrateq{\star}\strat'$.

  \noindent\textbf{Proof of (\ref{item:1}):}
  We take
  $c\in\choicep{p}(s)$, 
  and show that
  $\valh{\star}{\hist s_{(a,c)}}{p}\leq \valh{\star}{\hist
    s_{(b,c)}}{p}$.  For every $\tau\in\rectset_{-p}$, if
  $(\strat,\tau)\starmodels{\Phi(p)}$ then
  $(\strat',\tau)\starmodels{\Phi(p)}$.  In particular, if
  $(\strat,\tau)\starmodelsh{\hist s_{(a,c)}}{\Phi(p)}$ then
  $(\strat',\tau)\starmodelsh{\hist s_{(b,c)}}{\Phi(p)}$.  We deduce
  that $\strat'$ is winning from $\hist s_{(b,c)}$ if $\strat$ is
  winning from $\hist s_{(a,c)}$.  As $\strat$ is arbitrary this shows
  that
  $\valh{\star}{\hist s_{(a,c)}}{p}=1\Rightarrow \valh{\star}{\hist
    s_{(b,c)}}{p}=1$.  Dually if $\valh{\star}{\hist s_{(b,c)}}{p}=-1$
  then $\strat'$ is losing from $\hist s_{(b,c)}$ and so is $\strat$
  from $\hist s_{(a,c)}$.  This shows the implication
  $\valh{\star}{\hist s_{(b,c)}}{p}=-1\Rightarrow \valh{\star}{\hist
    s_{(a,c)}}{p}=-1$.  These two implications yield
  $\valh{\star}{\hist s_{(a,c)}}{p}\leq \valh{\star}{\hist
    s_{(b,c)}}{p}$.

  \noindent\textbf{Proof of (\ref{item:2}):}
  We show the contrapositive, consider $c$ such that
  $s_{(a,c)}\neq s_{(b,c)}$ and show that
  $\valh{\star}{\hist s_{(a,c)}}{p}$ and
  $\valh{\star}{\hist s_{(b,c)}}{p}$ cannot be both equal to $0$.  For
  this purpose we assume that $\valh{\star}{\hist s_{(a,c)}}{p}=0$ and
  show that $\valh{\star}{\hist s_{(b,c)}}{p}=1$.  Since $\strat$ is
  chosen arbitrary we take it such that
  $\valhstrat{\star}{\hist s_{(a,c)}}{\strat}=\valh{\star}{\hist
    s_{(a,c)}}{p}=0$.  In particular there exists
  $\tau\in\rectset_{-p}$ such that
  $(\strat,\tau)\starmodelsh{\hist s_{(a,c)}}{\Phi(p)}$.  Let $\tau'$
  be an arbitrary player $-p$ strategy.
  By~$\strat \bothdomstrateq{\star} \strat'$, we have
  $(\strat,\tau') \models^\star \Phi(p)$ implies
  $(\strat', \tau') \models^\star \Phi(p)$.  Since
  $\switchstrat{\tau'}{\hist s_{(a,c)}}{\tau}$ and $\tau'$ are equal
  on histories incomparable to $\hist s_{(a,c)}$, it then holds that
  $(\strat',\tau')\starmodelsh{\hist s_{(b,c)}}{\Phi(p)}$ for every
  arbitrary profile $\tau'$. We have shown that $\strat'$ is winning
  from $\hist s_{(b,c)}$ and hence that
  $\valh{\star}{\hist s_{(b,c)}}{p}=1$ as claimed.

  \noindent\textbf{Direction (\ref{item:1})$\wedge$(\ref{item:2})
    $\Rightarrow a\bothdommoveeq{\star}{\hist}b$:}
  We assume that (\ref{item:1}) and (\ref{item:2}) hold and show that
  $a\bothdommoveeq{\star}{\hist} b$.  Let $\strat$ be such that
  $\hist\in\prefixoutc(\strat)$ and $\strat(\hist)=a$.  We define
  $\strat'$ such that $\strat'(\hist)=b$ and
  $\strat\bothdomstrateq{\star}\strat'$ in several steps.  We start
  with $\strat'=\strat$.  We set $\strat'(\hist)=b$.  For every
  $s'\in \succmove{s}{b}\setminus \succmove{s}{a}$ we take a strategy
  $\strat_{s'}$ such that
  $\valhstrat{\star}{\hist s'}{\strat_{s'}}=\valh{\star}{\hist
    s'}{p}{}$ and update $\strat'$ to
  $\switchstrat{\strat'}{\hist s'}{\strat_{s'}}$.  For every
  $s'\in \succmove{s}{b}\cap\succmove{s}{a}$ such that
  $\valhstrat{\star}{\hist s'}{\strat}<\valh{\star}{\hist s'}{p}$ we
  do the same operation.  Note that for all states
  $s'\in \succmove{s}{b}\cap\succmove{s}{a}$ such that
  $\valhstrat{\star}{\hist s'}{\strat}=\valh{\star}{\hist s'}{p}$ the
  strategies $\strat'$ and $\strat$ coincide from $\hist s'$.
  Now we show that $\strat\bothdomstrateq{\star}\strat'$.  We take an
  arbitrary $\tau\in\rectset_{-p}$ such that
  $(\strat,\tau)\starmodels{\Phi(p)}$ and show that
  $(\strat',\tau)\starmodels{\Phi(p)}$.  Assume that
  $\hist\in \prefixoutc(\strat,\tau)$, otherwise $(\strat,\tau)$ and
  $(\strat',\tau)$ behave exactly the same.  Let
  $\gamma=\tau(\hist)$. We show that
  $(\strat',\tau)\starmodelsh{\hist}{\Phi(p)}$.  This is equivalent to
  show that $(\strat',\tau)\starmodelsh{\hist s_{(b,c)}}{\Phi(p)}$ for
  every $c\in\supp(\gamma)$.  Let
  $c\in\supp(\gamma)$. 
  If $s_{(b,c)}=s_{(a,c)}$ and
  $\valhstrat{\star}{\hist s_{(b,c)}}{\strat}=\valh{\star}{\hist
    s_{(b,c)}}{p}$ then $\strat$ and $\strat'$ play the same from
  $\hist s_{(a,c)}=\hist s_{(b,c)}$ and so
  $(\strat',\tau)\starmodelsh{\hist s_{(b,c)}}{\Phi(p)}$.  If
  $s_{(b,c)}\neq s_{(a,c)}$ then by (\ref{item:2}) and (\ref{item:1}),
  either
  $\valh{\star}{\hist s_{(b,c)}}{p}=\valh{\star}{\hist
    s_{(a,c)}}{p}=-1$, or
  $\valh{\star}{\hist s_{(b,c)}}{p}=\valh{\star}{\hist
    s_{(a,c)}}{p}=1$, or
  $\valh{\star}{\hist s_{(b,c)}}{p}>\valh{\star}{\hist s_{(a,c)}}{p}$.
  The first case is not possible since~$\tau$ witnesses that
  $\valh{\star}{\hist s_{(a,c)}}{p}>0$. In the second case, by
  construction, we
  have~$\valhstrat{\star}{\hist s_{(b,c)}}{\strat'} = 1$,
  thus~$(\strat',\tau) \starmodelsh{\hist s_{(b,c)}}{\Phi(p)}$.
  Assume the third case, which means
  $\valh{\star}{\hist s_{(b,c)}}{p}>\valh{\star}{\hist s_{(a,c)}}{p}$.
  Then
  $\valhstrat{\star}{\hist s_{(b,c)}}{\strat'}=\valh{\star}{\hist
    s_{(b,c)}}{p}>\valh{\star}{\hist s_{(a,c)}}{p}\geq 0$ (the first
  equality is due to the construction of $\strat'$, the last
  inequality is due to the fact that
  $(\strat,\tau)\starmodelsh{\hist s_{(a,c)}}{\Phi(p)}$).  We deduce
  that $\strat'$ is winning from $\hist s_{(b,c)}$ and hence that
  $(\strat',\tau)\starmodelsh{\hist s_{(b,c)}}{\Phi(p)}$.  This
  concludes the proof of $(\strat',\tau)\starmodelsh{\hist}{\Phi(p)}$.
  As $\strat'$ agrees with $\strat$ on histories for which $\hist$ is
  not a prefix, this implies that
  $(\strat',\tau)\starmodels{\Phi(p)}$.
\end{proof}

Now, we turn our attention to randomised moves in general, and show,
again, how the notion of $\star$-LA equivalence between
\emph{$\star$-LA moves} relates to values of histories. Let $\hist$ be
a history ending in $s$, and assume that player 1 has the choice
between two randomised moves $\alpha$ and $\beta$ which are $\star$-LA
\emph{and} equivalent ($\alpha\botheqadm{\star}{\hist}\beta$). Then,
the probability of ending up in some state $s'$ will be the same by
playing $\alpha$ or $\beta$, provided $\hist s'$ has value $0$ (for
player $p$). A similar property exists for histories of value $1$: the
probability of producing a history of value $1$ is the same under
$\alpha$ and $\beta$:

\begin{lemma}\label{lem:samedistronestep}
  Let $\hist$ be a history s.t. $\last{\hist}=s$, and let $\alpha$ and
  $\beta$ be two randomised moves in $\distr{\choicep{p}(s)}$
  s.t. $\alpha\botheqadm{\star}{\hist}\beta$. Then, for all $\gamma\in \distr{\choicep{-p}(s)}$:
\begin{enumerate}[(i)]
\item \label{item:5}
  $\transfunrand(s,(\alpha, \gamma))(s')=\transfunrand(s,(\beta,
  \gamma))(s')$ for all $s'$ s.t.
  $\hist s'\in \bothvalpred{\star}{p}{0}$; and
\item \label{item:6}
  $\transfunrand(s,(\alpha,
  \gamma))(S_1)=\transfunrand(s,(\beta,
  \gamma))(S_1)$, with
  $S_1=\{s'\mid \hist s'\in \bothvalpred{\star}{p}{1}\}$.
\end{enumerate} 
\end{lemma}

\begin{proof}
  We first prove the properties for Dirac moves $\alpha=a$ and
  $\beta=b$ such that $a\botheqadm{\star}{\hist}b$.  By
  Lemma~\ref{lem:dominemove},
  $\valhist{\hist \transfun(s,a,c)}{p}=\valhist{\hist
    \transfun(s,b,c)}{p}$ for every $c\in\choicep{-p}(s)$ and if this
  value is $0$ it further holds that:
  $\transfun(s,a,c)=\transfun(s,b,c)$.  Then for every $s'$ such that
  $\valhist{\hist s'}{p}=0$, for every $c\in \supp(\gamma)$, it holds
  that $\transfunrand(s,(a, c))(s')=\transfunrand(s,(b, c))(s')$.
  Thus (\ref{item:5}) holds when $\alpha=a$ and $\beta=b$ are Dirac
  moves:
  \begin{align*}
    \transfunrand(s,(a, \gamma))(s')&=\sum_{c\in
                                      \supp(\gamma)}\gamma(c) 
                                      \transfunrand(s,(a, c))(s')\\
                                    &=\sum_{c\in \supp(\gamma)} \gamma(c)
                                      \transfunrand(s,(b, c))(s')\\
                                    &=\transfunrand(s,(b, \gamma))(s').
  \end{align*}

  Next, we know that, for all Dirac moves $c\in\choicep{-p}(s)$,
  $\transfunrand(s,(a, c))(\transfun(s,a,c))=\transfunrand(s,(b,
  c))(\transfun(s,b,c)) = 1$, by definition of Dirac moves. We use this to
  establish point~(\ref{item:6}), again when $\alpha=a$ and $\beta=b$
  are Dirac moves:
  \begin{align*}
    \transfunrand(s,(a,  \gamma))(S_1) 
    &= \sum_{s' \mid \hist s'\in \bothvalpred{\star}{p}{1}} \transfunrand(s,(a,\gamma))(s')\\
    &=\sum_{c\mid \hist \transfun(s,a,c) \in \bothvalpred{\star}{p}{1}} 
      \gamma(c) \transfunrand(s,(a,c))(\transfun(s,a,c))\\
    &=\sum_{c\mid \hist \transfun(s,b,c) \in \bothvalpred{\star}{p}{1}}
      \gamma(c) \transfunrand(s,(b,c))(\transfun(s,b,c))\\
    &=\sum_{s' \mid \hist s'\in \bothvalpred{\star}{p}{1}}
      \transfunrand(s,(b,\gamma))(s')\\
    &=\transfunrand(s,(b,  \gamma))(S_1) .
  \end{align*}
  Notice that the third inequality follows from $a \botheqadm{\star}{\hist} b$ and Lemma~\ref{lem:dominemove}, since
  we have $\{c \mid h\delta(s,a,c)  \in \bothvalpred{\star}{p}{1}\}=
   \{c \mid h\delta(s,b,c)  \in \bothvalpred{\star}{p}{1}\}$.
  \medskip
 
  To finish the proof, we establish both properties for general
  randomised moves $\alpha$ and $\beta$.  We know that $\alpha$'s and
  $\beta$'s supports are equivalent Dirac moves by
  Lemma~\ref{lem:characadmmoves}. We start by selecting some
  $b\in\supp(\beta)$, in order to reuse the results above. We have:
  \begin{align*}
    \transfunrand(s,(\alpha, \gamma))(s') 
    &= \sum_{a\in \supp(\alpha)}\transfunrand(s,(a, \gamma))(s')\alpha(a)\\
    &= \sum_{a\in \supp(\alpha)}\transfunrand(s,(b,
      \gamma))(s')\alpha(a)
      &\text{as proved above}\\
    &=\transfunrand(s,(b, \gamma))(s') \sum_{a\in \supp(\alpha)}\alpha(a)\\
    &= \transfunrand(s,(b, \gamma))(s').
  \end{align*}
  So, summing up:
  \begin{align}
    \transfunrand(s,(\alpha, \gamma))(s')  &= \transfunrand(s,(b, \gamma))(s').\label{eq:1}
  \end{align}
  However:
  \begin{align*}
    \transfunrand(s,(\beta, \gamma))(s') 
    &= \sum_{b\in \supp(\beta)}\transfunrand(s,(b,
      \gamma))(s')\beta(b)\\
    &= \sum_{b\in \supp(\beta)} \transfunrand(s,(\alpha, \gamma))(s')
      \beta(b)
    &\text{by~\eqref{eq:1}}\\
    &= \transfunrand(s,(\alpha, \gamma))(s') \sum_{b\in
      \supp(\beta)}\beta(b)\\
    &= \transfunrand(s,(\alpha, \gamma))(s').
  \end{align*}
  Therefore establishing (\ref{item:5}).
  
  Finally, we proceed similarly for $S_1$. Again, we start by fixing
  $b\in\supp(\beta)$. We have:
  \begin{align*}
    \transfunrand(s,(\alpha,  \gamma))(S_1) 
    &= \sum_{s' \mid \hist s'\in \bothvalpred{\star}{p}{1}}
      \transfunrand(s,(\alpha,\gamma))(s')\\
   &= \sum_{a\in\supp(\alpha)} \left(\alpha(a)
    \sum_{s' \mid \hist s'\in \bothvalpred{\star}{p}{1}}
      \transfunrand(s,(a,\gamma))(s')\right)\\
    &= \sum_{a\in\supp(\alpha)} \left(\alpha(a)
    \sum_{s' \mid \hist s'\in \bothvalpred{\star}{p}{1}}
      \transfunrand(s,(b,\gamma))(s')\right)
   &\text{as proved above}\\
    &=\left(\sum_{s' \mid \hist s'\in \bothvalpred{\star}{p}{1}}
      \transfunrand(s,(b,\gamma))(s')\right)
      \left(\sum_{a\in\supp(\alpha)} \alpha(a)\right)\\
    &= \sum_{s' \mid \hist s'\in \bothvalpred{\star}{p}{1}}
      \transfunrand(s,(b,\gamma))(s').
  \end{align*}
  So, summing up:
  \begin{align}
    \transfunrand(s,(\alpha,  \gamma))(S_1) &= \sum_{s' \mid \hist s'\in \bothvalpred{\star}{p}{1}}
      \transfunrand(s,(b,\gamma))(s').\label{eq:2}
  \end{align}
  However,
  \begin{align*}
    \transfunrand(s,(\beta,  \gamma))(S_1) 
    &= \sum_{s' \mid \hist s'\in \bothvalpred{\star}{p}{1}}
      \transfunrand(s,(\beta,\gamma))(s')\\
    &= \sum_{b\in\supp(\beta)} \left(\beta(b)
      \sum_{s' \mid \hist s'\in \bothvalpred{\star}{p}{1}}
      \transfunrand(s,(b,\gamma))(s')\right)\\
    &=\sum_{b\in\supp(\beta)} \beta(b) \transfunrand(s,(\alpha,
      \gamma))(S_1) 
    &\text{by~\eqref{eq:2}}\\
    &= \transfunrand(s,(\alpha,
      \gamma))(S_1) \sum_{b\in\supp(\beta)}\beta(b)\\
    &=  \transfunrand(s,(\alpha, \gamma))(S_1) 
  \end{align*}
  Therefore establishing (\ref{item:6}).
\end{proof}

\subparagraph*{Characterisation and existence of admissible
  strategies} Equipped with our previous results, we can now establish
the main results of this section. First, we show that
\emph{$\star$-admissible strategies are exactly those that are both
  $\star$-LA and $\star$-SCO} (Theorem~\ref{theo:classifconc}$(i)$). Then,
we show that \emph{admissible strategies always exist} in concurrent
games (Theorem~\ref{theo:exist-adm}$(ii)$).

To establish Theorem~\ref{theo:exist-adm}, we need an ancillary lemma,
which relates the notions of $\star$-admissibility and the notions of
$\star$-LA. To this end, we draw a link between (global) equivalence
of strategies (in terms of $\botheqadmstrat{\star}$) and the local
equivalence of moves (in terms of $\botheqadm{\star}{\hist}$). For all
player $p$ strategies $\strat$ and all $v\in\{-1,0,1\}$, let us define
$\prefixoutc_v(\strat)$ as
$\prefixoutc(\strat)\cap \bothvalpred{\star}{p}{v}$, i.e., the set of
histories of $\strat$ that have value $v$. Then:

\begin{lemma}\label{lem:equivstrat}
  Let $\strat$ and $\strat'$ be two player $p$ strategies s.t.:
  \begin{inparaenum}[(i)]
  \item $\strat$ is $\star$-LA; and
  \item for all histories $\hist$: $\valh{\star}{\hist}{p}=1$ implies
    that $\strat$ and $\strat'$ are $\star$-winning from $\hist$.
  \end{inparaenum}
  Then, the following conditions are equivalent:
  \begin{enumerate}[(1)]
  \item\label{it1:equivstrat}
    $\strat(\hist)\botheqadm{\star}{\hist}\strat'(\hist)$ for all
    $\hist\in\prefixoutc(\strat)\cap\prefixoutc(\strat')$;
  \item\label{it2:equivstrat}
    $\prefixoutc_0(\strat)=\prefixoutc_0(\strat')$ and
    $\strat(\hist)\botheqadm{\star}{\hist}\strat'(\hist)$ for all
    $\hist\in\prefixoutc_0(\strat)$;
  \item\label{it3:equivstrat} $\strat\botheqadmstrat{\star}\strat'$
  \end{enumerate}
\end{lemma}

\begin{proof}
  \noindent\textbf{Proof of $(\ref{it1:equivstrat})\Rightarrow(\ref{it2:equivstrat})$ :}
  By induction and application of Lemma \ref{lem:samedistronestep}.

  \noindent\textbf{Proof of  $(\ref{it2:equivstrat})\Rightarrow(\ref{it3:equivstrat})$}
  We assume that
  $\forall \hist\in\prefixoutc_0(\strat)=\prefixoutc_0(\strat')$,
  $\strat(\hist)\botheqadm{\asure}{\hist}\strat'(\hist)$ and show that
  $\strat\botheqadmstrat{\asure}{}\strat'$.

  It suffices to take an arbitrary $\tau\in \rectset_{-p}$ and show
  that $\Prob_{(\strat,\tau)}(\Phi(p))=1$ iff
  $\Prob_{(\strat',\tau)}(\Phi(p))=1$.  We first show that:
  \begin{equation}\label{eq:equalcylinders}
    \forall \hist\in\prefixoutc_0(\strat)=\prefixoutc_0(\strat'),\quad \Prob_{(\strat,\tau)}(\cyl{\hist})=\Prob_{(\strat',\tau)}(\cyl{\hist}).
  \end{equation}
  This can be shown by induction using Lemma
  \ref{lem:samedistronestep} where the base case is
  $\Prob_{(\strat,\tau)}(\cyl{\vinit})=1=\Prob_{(\strat',\tau)}(\cyl{\vinit})$
  and the induction step is that
  $\Prob_{(\strat,\tau)}(\cyl{\hist})=\Prob_{(\strat',\tau)}(\cyl{\hist})$
  implies:
  \begin{align*}
  \Prob_{(\strat,\tau)}(\cyl{\hist'
    s'}) &=\Prob_{(\strat,\tau)}(\cyl{\hist'})
  \transfunrand(s,(\strat(\hist'),
  \tau(\hist')))(s')\\
    &=\Prob_{(\strat',\tau)}(\cyl{\hist'})
  \transfunrand(s,(\strat'(\hist'),
  \tau(\hist')))(s')\\
    &=\Prob_{(\strat,\tau)}(\cyl{\hist'
    s'})
  \end{align*}

  We decompose the set of runs as
  $\Box \bothvalpred{\star}{p}{0} \biguplus \Diamond
  \bothvalpred{\star}{p}{1}\biguplus \Diamond
  \bothvalpred{\star}{p}{-1}$.  When intersecting with $\Phi(p)$, the
  last set in the union becomes empty and we have:
  \[
    \Phi(p)=\left(\Box \bothvalpred{\star}{p}{0}\cap \Phi(p)\right)
    \biguplus \left(\Diamond \bothvalpred{\star}{p}{1}\cap
      \Phi(p)\right).
  \]

  Hence for $\strat''\in\{\strat,\strat'\}$,
  $\Prob_{(\strat'',\tau)}(\Phi(p))=\Prob_{(\strat'',\tau)}(\Box
  \bothvalpred{\star}{p}{0}\cap
  \Phi(p))+\Prob_{(\strat'',\tau)}(\Diamond
  \bothvalpred{\star}{p}{1})$ where we used that
  $\Diamond \bothvalpred{\star}{p}{1}$ implies $\Phi(p)$ almost surely
  because $\strat''$ is winning from histories of value $1$.

  Hence it suffices to show that
  $\Prob_{(\strat,\tau)}(\Box \bothvalpred{\star}{p}{0}\cap
  \Phi(p))=\Prob_{(\strat',\tau)}(\Box \bothvalpred{\star}{p}{0}\cap
  \Phi(p))$ and
  $\Prob_{(\strat,\tau)}(\Diamond
  \bothvalpred{\star}{p}{1})=\Prob_{(\strat',\tau)}(\Diamond
  \bothvalpred{\star}{p}{1})$.  The former equality is due to
  \eqref{eq:equalcylinders}.
  To prove the latter equality, we write for
  $\strat''\in\{\strat,\strat'\}$,
  \begin{equation}\label{eq:decompoone}
    \Prob_{(\strat'',\tau)}(\Diamond \bothvalpred{\star}{p}{1})=\sum_{\hist\in \prefixoutc_0(\strat'')}\Prob_{(\strat'',\tau)}(\cyl{\hist})\transfunrand(s,(\strat''(\hist),\tau(\hist)))(\bothvalpred{\star}{p}{1}).
  \end{equation}
  This equality traduces the decomposition of the even
  $\bothvalpred{\star}{p}{1}$ into the disjoint events parametrised by
  the history $\hist$ of value $0$ just before entering
  $\bothvalpred{\star}{p}{1}$.  Using \eqref{eq:decompoone}, Lemma
  \ref{lem:samedistronestep} and \eqref{eq:equalcylinders} we deduce
  that
  $\Prob_{(\strat,\tau)}(\Diamond
  \bothvalpred{\star}{p}{1})=\Prob_{(\strat',\tau)}(\Diamond
  \bothvalpred{\star}{p}{1})$, ending the proof.

  \noindent\textbf{Proof of
    $(\ref{it3:equivstrat})\Rightarrow(\ref{it1:equivstrat})$:
  }Straightforward.
\end{proof}

\begin{theorem}[Characterisation and existence of admissible
  strategies]\label{theo:classifconc}\label{theo:exist-adm}
  The following holds for all strategies $\strat$ in a concurrent game with semantics
  $\star\in\{\sure,\asure\}$: 
  \begin{enumerate}[(i)]
  \item \label{item:7} $\strat$ is $\star$-admissible iff $\strat$ is $\star$-LA and
    $\star$-SCO; in the special case of simple safety objectives, if $\strat$ is $\star$-LA then $\strat$ is $\star$-admissible.
  \item \label{item:8} there is a $\star$-admissible strategy $\strat'$ such that
    $\strat\bothdomstrateq{\star}\strat'$.
  \end{enumerate}
  In particular, point~(\ref{item:8}) implies that admissible strategies
  always exist in concurrent games.
\end{theorem}
\begin{proof}[Proof of Theorem~\ref{theo:exist-adm}]
  We start with the proof of point~(\ref{item:7}).
  
  \noindent\textbf{"$\strat$ is $\star$-admissible" implies "$\strat$ is
    $\star$-SCO":}  
  We show the contrapositive. Let $\strat$ be a strategy
  that is not $\star$-SCO, and let us show that $\strat$ is not
  $\star$-admissible. Since $\strat$ is not $\star$-SCO, there is, by
  definition of $\star$-SCO, a history $\hist$
  s.t. $\valhstrat{\star}{\hist}{\strat}\neq\valh{\star}{\hist}{p}$. Recall
  that
  $\valh{\star}{\hist}{p}=\max_{\strat'}\valhstrat{\star}{\hist}{\strat'}$. Hence,
  $\valhstrat{\star}{\hist}{\strat}\neq\valh{\star}{\hist}{p}$ implies
  that
  $\valhstrat{\star}{\hist}{p}{\game}{\strat}<\valh{\star}{\hist}{p}$
  and that there exists a strategy $\strat'$
  s.t. $\valhstrat{\star}{\hist}{\strat}=\valh{\star}{\hist}{p}$. Consider
  the strategy
  $\overline{\strat}=\switchstrat{\strat}{\hist}{\strat'}$ that plays
  like $\strat$ and switches to $\strat'$ after history $\hist$. We
  claim that $\strat \bothdomstrat{\star} \overline{\strat}$.  Observe
  that, by construction:
  $\valhstrat{\star}{\hist}{\strat} <
  \valhstrat{\star}{\hist}{\overline{\strat}}$.  Note also that
  $\valhstrat{\hist}{\strat}=-1$ is not possible, because this means
  that all continuations of $\hist$ are losing for $p$, so it is not
  possible to have $\overline{\strat}$ with
  $\valhstrat{\star}{\hist}{\strat} <
  \valhstrat{\star}{\hist}{\overline{\strat}}$. So we have necessarily
  $\valhstrat{\hist}{\strat}=0$ and
  $\valhstrat{\hist}{\overline{\strat}}=1$. This means that
  $\overline{\strat}$ is winning from $h$ (against all possible
  strategies of $-p$), while, by definition of
  $\valhstrat{\star}{\hist}{\strat}=0$, $\strat$ is, from $\hist$,
  losing for $p$ against at least one strategy of $-p$. Hence,
  $\strat \bothdomstrat{\star} \overline{\strat}$ and $\strat$ is not
  $\star$-admissible.

  \noindent\textbf{"$\strat$ is $\star$-admissible" implies "$\strat$
    is $\star$-LA":}
  We show the contrapositive: we assume that $\strat$ is not
  $\star$-LA and show that $\strat$ is not $\star$-admissible.  There
  exists a history $\hist\in\prefixoutc(\strat)$ such that
  $\strat(\hist)$ is $\star$-dominated at $\hist$ by an $\star$-LA
  move $b$ that can be chosen Dirac by virtue of Lemma
  \ref{lem:characadmmoves} (\ref{item:4}).  There exists
  $a\in\supp(\strat(h))$ such that $a\bothdommove{\star}{\hist} b$
  (otherwise all the moves of the support of $\supp(\strat(h))$ would
  be equivalent to the $\star$-LA move $b$ and so would be $\alpha$ by
  Lemma \ref{lem:characadmmoves} (\ref{item:3})).  We already saw in
  the proof of Lemma \ref{lem:Diracdominates} (\ref{item:4}), that
  $\strat\bothdomstrateq{\star}\strat_a$ where $\strat_a$ is the
  strategy that plays like $\strat$ everywhere except in $\hist$ where
  it plays $a$ instead of $\alpha$.  Then it suffices to show that
  $\strat_a$ is $\star$-dominated.  Using Lemma \ref{lem:dominemove}
  and notation therein with $a\bothdommoveeq{\star}{\hist} b$, we know
  that for every $c\in \choicep{-p}(s)$,
  $\valh{\star}{\hist s_{(a,c)}}{p}\leq \valh{\star}{\hist
    s_{(b,c)}}{p}$ and if
  $\valh{\star}{\hist s_{(a,c)}}{p}= \valh{\star}{\hist
    s_{(b,c)}}{p}=0$ then $s_{(a,c)}=s_{(b,c)}$.  Here we can reuse
  the end of the proof of Lemma \ref{lem:dominemove} to find a
  strategy $\strat'$ such that $\strat_a\bothdomstrateq{\star}\strat'$
  and that satisfies the further properties that it is $\star$-SCO at
  $\hist s_{(b,c)}$ for every $c\in \choicep{-p}(s)$, and it plays as
  $\hist$ in every proper prefix of $\hist$.  It remains to show that
  there exists a profile $\tau\in\rectset_{-p}$ such that
  $(\strat_a,\tau)\not \starmodels{\Phi(p)}$ and
  $(\strat',\tau)\starmodels{\Phi(p)}$.  Using Lemma
  \ref{lem:dominemove}, this time with
  $\strat_a(\hist)\not \bothdommoveeq{\star}{\hist} a$ we get the
  existence of a Dirac move $c\in \choicep{-p}(s)$ such that
  $\valh{\star}{\hist s_{(a,c)}}{p}< \valh{\star}{\hist
    s_{(b,c)}}{p}$.  We define $\tau\in\rectset_{-p}$ such that
  $\hist\in\prefixoutc(\strat_a,\tau)$, $\tau(\hist)=c$ and a further
  property depending on cases described as follows.  If
  $\valhist{\hist s_a}{p}=-1$ then $\valhist{\hist s_b}{p}\geq 0$ and
  we further require that
  $(\strat',\tau)\starmodelsh{\hist s_b}{\Phi(p)}$.  If
  $\valhist{\hist s_a}{p}=0$ then $\valhist{\hist s_b}{p}=1$ and we
  further require that
  $(\strat_a,\tau)\not\starmodelsh{\hist s_a}{\Phi(p)}$.  In both
  cases it holds that $(\strat',\tau)\starmodelsh{\hist s_b}{\Phi(p)}$
  and $(\strat_a,\tau)\not\starmodelsh{\hist s_a}{\Phi(p)}$.  So
  $(\strat_a,\tau)\not\starmodels{\Phi(p)}$ and
  $(\strat',\tau)\starmodels{\Phi(p)}$.  We conclude that $\strat$ is
  not $\star$-LA because
  $\strat\bothdomstrateq{\star}\strat_a\bothdomstrat{\star}\strat'$.

  \noindent\textbf{"$\strat$ is $\star$-SCO and $\star$-LA" implies
    "$\strat$ is $\star$-admissible":}
  Let $\strat$ be a strategy that is both $\star$-LA and $\star$-SCO.
  We show that for every $\strat'$ such that
  $\strat\bothdomstrateq{\star}\strat'$ implies
  $\strat\botheqadmstrat{\star}\strat'$; which is a way of proving
  that $\strat$ is $\star$-Admissible.  Let $\strat'\in\rectset_p$
  such that $\strat\bothdomstrateq{\star}\strat'$.  We can assume
  without loss of generality that $\strat'$ is winning from histories
  of value $1$ (otherwise it suffices to change $\strat'$ so that when
  a history of value $1$ it switches to a winning strategy from that
  history).  $\strat\bothdomstrateq{\star}\strat'$ we deduce that
  $\strat(\hist) \bothdommoveeq{\star}{\hist} \strat'(\hist)$ for
  every $\hist\in \prefixoutc(\strat)\cap \prefixoutc(\strat')$, then
  $\strat(\hist) \eqadm{\hist} \strat'(\hist)$ because $\strat(\hist)$
  is $\star$-LA at $\hist$.  We can apply Lemma \ref{lem:equivstrat},
  and deduce that $\strat\bothdomstrateq{\star}\strat'$ as required.

  \noindent\textbf{In simple safety games: "$\strat$ is $\star$-LA" implies
    "$\strat$ is $\star$-admissible":} To establish this case, we
  prove that, in simple safety games, "$\strat$ is $\star$-LA"
  implies "$\strat$ is $\star$-SCO". This implies that all $\star$-LA
  strategies are also $\star$-SCO, hence admissible. First, observe
  that, in the case of simple safety games, the two semantics $\asure$
  and $\sure$ coincide. Indeed, a winning strategy for the $\sure$
  semantics is also winning for the almost sure semantics. On the
  other hand, a winning strategy for the almost sure semantics ensures
  that no prefix generated by the strategy reaches the bad states
  (otherwise, this prefix would have non-negative measure, and the
  probability of winning would be $<1$). Hence, a winning strategy for
  the almost sure semantics is also winning for the sure
  semantics. So, we can restrict ourselves to the sure semantics in
  the arguments that follow.

  From $\strat$, we build a new Dirac $\star$-LA strategy $\strat'$,
  s.t., for all histories $\hist$, $\strat'(\hist)$ is an arbitrarily
  chosen action from $\supp(\strat(\hist))$. By
  Lemma~\ref{lem:Diracdominates}, point~(\ref{it3:characadmmoves}),
  $\strat'$ is also $\star$-LA. We claim that $\strat'$ is
  \emph{value-preserving} i.e., for all histories $\hist$:
  $\valh{\sure}{\hist}{p}\leq\valh{\sure}{\hist s}{p}$, for all
  $s\in \succmove{\last{\hist}}{\strat'(\hist)}$ (assuming $\strat$,
  and hence $\strat'$, are player $p$ strategies). In other words, by
  playing $\sigma'$, player $p$ never decreases the value of the
  histories he visits. This value preserving property holds by
  Lemma~\ref{lem:dominemove}.

  Moreover, since the objective is a simple safety one, it is easy to
  see that when $\strat'$ is $\sure$-winning, then $\strat$ is
  $\sure$-winning too. So $\strat'$ is dominated by $\strat$ and it is
  sufficient to prove that $\strat'$ is admissible to show that
  $\strat$ is admissible too.

  So, when states of value $1$ are entered, $\strat'$ ensures that
  they are never left, and $\strat'$ is thus winning from states of
  value $1$. On the other hand, when a history of value $0$ is
  entered, $\strat'$ will also ensure that only histories of value $0$
  or $1$ are visited. In both cases, the $\strat'$ is winning because
  we are considering a \emph{safety objective}. In particular, if we
  stay in histories of value $0$ forever, the strategy is winning
  because the bad states are never visited (otherwise, the value of
  the history would drop to $-1$).

  We conclude that $\strat'$ is $\star$-SCO. Since it is also
  $\star$-LA, $\strat'$ is admissible. Hence, $\strat$ is admissible too.

  \medskip

  Now, we move to the proof of point~(\ref{item:8}).
  
  Let $\strat$ be a strategy for player $p$, we build an admissible
  strategy that weakly dominates $\strat$.  This proof is similar to
  the proof of a similar result for SCO strategies in turn based games
  (\cite[Lemma 7]{FSTTCS2016}).  The strategy $\strat'$ is
  built dynamically along runs. It plays a winning strategy as soon as
  the current history begins to be of value $1$.  When the history
  begins to be of value $-1$, there is nothing more to do, the
  strategy can play
  arbitrarily. 
  Consider now, runs in which the value is always $0$. The strategy
  $\strat'$ keeps in memory a wished strategy $\strat_{\hist}$ for
  himself and a wished strategy $\tau_{\hist}$ of player $-p$ such
  that $(\strat_{\hist},\tau_{\hist})\starmodelsh{\hist}{\Phi(p)}$.
  The strategies $\strat_{\hist}$ and $\tau_{\hist}$ are inductively
  defined as follows.  At the beginning, $\hist=\vinit$,
  $\strat_{\hist}$ and $\tau_{\hist}$ are chosen such that
  $(\strat_{\hist},\tau_{\hist})\starmodelsh{\hist}{\Phi(p)}$ and such
  that $\strat_{\hist}$ is $\hist$-admissible. Here and below if
  $\strat$ satisfies these properties then the choice
  $\strat_{\hist}=\strat$ is made by default, otherwise another
  strategy $\strat_{\hist}$ is chosen (and it dominates $\strat$).
  After history $\hist$, $\strat'$ plays $\strat_\hist(\hist)$.  If
  after some history $\hist$ the continuation $\hist s'$ chosen is not
  in $\prefixoutcome{\strat_{\hist},\tau_{\hist}}$ (this is the case
  only if player $-p$ plays something else than
  $\tau_{\hist}(\hist)$), then $\strat_{\hist s'}$ and
  $\tau_{\hist s'}$ are chosen such that $\strat_{\hist s'}$ is
  $\hist s'$-admissible and
  $(\strat_{\hist s'},\tau_{\hist s'})\starmodelsh{\hist
    s'}{\Phi(p)}$.  Otherwise, if the wished profile is followed then
  the wished strategy for player $-p$ is left unchanged
  ($\tau_{\hist s'}=\tau_{\hist}$); the wished strategy for player $p$
  is left unchanged ($\strat_{\hist s'}=\strat_{\hist}$) if
  $\strat_{\hist}$ is $\hist s'$-admissible and otherwise
  $\strat_{\hist s'}$ is defined as a strategy that plays an
  admissible moves in $\hist s'$ and that dominates $\strat_{\hist}$.
  The constructed strategy is $\star$-LA.  It is also $\star$-SCO
  because from every history $\hist$ of value $0$,
  $(\strat',\tau_{\hist})\starmodelsh{\hist}{\Phi(p)}$ and from every
  history of value $1$, $\strat'$ is winning.  $\strat'$ plays as
  $\strat$ by default or a strategy that dominates $\strat$. At the
  end $\strat'$ is an admissible strategy that weakly dominates
  $\strat$.
\end{proof}

\begin{example}
  We consider again the example in \figurename~\ref{ex:runex}, and
  consider strategies $\strat_2$ and $\strat_3$ as defined in
  Example~\ref{example:running}. Remember that these two strategies do
  their best to reach $s_2$, and that, from $s_2$, $\strat_2$ plays
  deterministically $f$, while $\strat_3$ plays $f$ and $g$ with equal
  probabilities. From Example~\ref{example:sco}, we know that
  $\strat_2$ is $\sure$-SCO but not $\asure$-SCO; while $\strat_3$ is
  $\asure$-SCO but not $\sure$-SCO. Indeed, we have already argued in
  Example~\ref{example:admissible} that $\strat_2$ is not
  $\asure$-admissible, and that $\strat_3$ is not
  $\sure$-admissible. However, from
  Example~\ref{example:locally-admissible}, we know that $\strat_2$ is
  $\sure$-LA and that $\strat_3$ is $\asure$-LA. So, by
  Theorem~\ref{theo:classifconc}, $\strat_2$ is $\sure$-admissible and
  $\strat_3$ is $\asure$-admissible as expected.
\end{example}

Finally, we close the section by a finer characterisation of
$\star$-admissible strategies. We show that:
\begin{inparaenum}[(i)]
\item in the sure semantics, there is always an $\sure$-admissible
  strategy that plays Dirac moves only; and
\item in the almost-sure semantics, there is always an $\asure$-admissible
  strategy that plays Dirac moves only in histories of values $0$ or $-1$.  
\end{inparaenum}
The difference between the two semantics should not be surprising, as
we know already that randomisation is sometimes needed to win (i.e.,
from histories of value $1$) in the almost sure semantics:
\begin{proposition}\label{theo:Diracsuffices}
  For all player $p$ strategies $\strat$ in a concurrent game with
  $\star\in\{\sure,\asure\}$:
  \begin{enumerate}[(i)]
  \item If $\strat$ is $\asure$-admissible then there exists a
    strategy $\strat'$ that plays only Dirac moves in histories of
    value $\leq 0$ such that $\strat\botheqadm{\asure}{}\strat'$.
  \item If $\strat$ is $\sure$-admissible then there exists a Dirac
    strategy $\strat'$ such that $\strat\botheqadm{\sure}{}\strat'$.
  \end{enumerate}
\end{proposition}
\begin{proof}
  Take $\strat$ is $\asure$-admissible.  We define $\strat'$ as
  follows.  From histories of value $-1$, plays arbitrary Dirac
  strategies.  From histories of value $1$, plays winning
  strategies. If $\star=\sure$ we can further requires that such
  winning strategies are Dirac.  If there exists histories of value
  $0$ then necessarily the initial state $\vinit$ is such a history
  and we build $\strat'$ inductively as follows such that
  $\prefixoutc_0(\strat')=\prefixoutc_0(\strat)$.  Consider
  $\hist \in \prefixoutc_0(\strat)$, then define $\strat'(\hist)$ to
  be an arbitrary Dirac move $a\in\supp(\strat(\hist))$.  Then by
  Lemma \ref{lem:characadmmoves} (\ref{it3:characadmmoves}) and the
  fact that $\strat(\hist)$ is $\star$-LA we deduce that
  $\strat'(\hist)\botheqadm{\star}{\hist}\strat(\hist)$; The histories
  of the form $\hist s'$ that are in $\prefixoutc_0(\strat')$ are
  exactly those that belongs to $\prefixoutc_0(\strat)$ (see Lemma
  \ref{lem:samedistronestep}).  By induction,
  $\prefixoutc_0(\strat')=\prefixoutc_0(\strat)$ and for every
  $\hist\in\prefixoutc_0(\strat)$,
  $\strat'(\hist)\botheqadm{\star}{\hist}\strat(\hist)$.  We can use
  Lemma \ref{lem:equivstrat} and deduce that
  $\strat\botheqadm{\star}{\hist}\strat'$ with $\strat'$ Dirac in
  histories of value $\leq 0$, and $\leq 1$ when $\star=\asure$.
\end{proof}


%% file: aasynth.tex
\section{Assume admissible synthesis}\label{sec:assume-admiss-synth}
In this section we discuss an \emph{assume-admissible synthesis}
framework for concurrent games. With classical synthesis, one tries
and compute \emph{winning} strategies for all players, i.e.,
strategies that \emph{always win} against \emph{all possible
  strategies} of the other players.  Unfortunately, it might be the
case that such \emph{unconditionally} winning strategies do not exist,
as in our example.  As explained in the introduction, the
assume-admissible synthesis rule relaxes the classical synthesis rule:
instead of searching from strategies that win unconditionally, the new
rule requires winning against the {\em admissible strategies} of the
other players. So, a strategy may satisfy the new rule while not
winning unconditionally.  Nevertheless, we claim that winning against
admissible strategies is good enough assuming that the players are
\emph{rational}; if we assume that players only play strategies that
are good for achieving their objectives, they will be playing
admissible ones.

The general idea of the assume-admissible synthesis algorithm is to
reduce the problem (in a concurrent $n$-player game) to the synthesis
of a winning strategy \emph{in a $2$-player zero-sum concurrent game
  of imperfect information, in the $\sure$-semantics} (even when the
original assume-admissible problem is in the $\asure$-semantics),
where the objective of player $1$ is given by an LTL formula.  Such
games are solvable using techniques presented
in~\cite{DBLP:journals/tocl/ChatterjeeAH11}.

More precisely, from a concurrent game $\game$ in the semantics
$\star\in\{\sure,\asure\}$ and player $p$, we build a game
$\bothgameadmp{\star}{p}$ with the above characteristics, which is used to 
decide the assume-admissible synthesis rule.
If such a solution exists, our algorithm  constructs a witness strategy. 
For example, the game
$\bothgameadmp{\star}{1}$ corresponding to the game in
\figurename~\ref{ex:runex} is given in
\figurename~\ref{fig:imperfectinfo}. The main ingredients for this
construction are the following.
\begin{inparaenum}[(i)]
\item In $\bothgameadmp{\star}{p}$, the protagonist is player $p$, and
  the second player is $-p$.
\item Although randomisation is needed to win in such games in
  general, we interpret $\bothgameadmp{\star}{p}$ in the
  $\sure$-semantics only.  In fact, we have seen that for the
  protagonist, Dirac moves suffice in states of value $0$; so the only
  states where he might need randomisation are those of value~$1$
  (randomisation does not matter if the value is~$-1$ since the
  objective is lost anyway).
  Hence we define winning condition to be
  $\Phi(p)\vee \Diamond \bothvalpredstates{\star}{p}{1}$ enabling us
  to consider only histories of values $0$ in
  $\bothgameadmp{\star}{p}$; and thus hiding the parts of the game
  where randomisation might be needed.
  We also prove that we can restrict to Dirac strategies for $-p$ when
  it comes to admissible strategies.
%
\item In order to restrict the strategies to admissible ones, we only
  allow~$\star$-LA moves in $\bothgameadmp{\star}{p}$.  These moves
  can be computed by solving classical $2$-player games
  (\cite{DBLP:journals/jacm/AlurHK02}) using
  Lemma~\ref{lem:dominemove}.
  For example, in \figurename~\ref{fig:imperfectinfo}, moves $c$ and
  $c'$ are removed since they are not $\asure$-LA.
\item Last, since $\star$-admissible strategies are those that are
  both $\star$-LA and $\star$-SCO (see
  Theorem~\ref{theo:classifconc}), we also need to ensure that the
  players play $\star$-SCO. This is more involved than $\star$-LA, as
  the $\star$-SCO criterion is not \emph{local}, and requires
  information about the \emph{sequence of actual moves} that have been
  played, which cannot be deduced, in a concurrent game, from the
  sequence of visited states. So, we store, in the states of
  $\bothgameadmp{\star}{p}$, the moves that have been played by all
  the players to reach the state. For example, in
  \figurename~\ref{fig:imperfectinfo}, the state labelled by
  $s_1, (b,b')$ means that $\game$ has reached $s_1$, and that the
  last actions played by the players were $b$ and $b'$ respectively.
  However, players' strategies must not depend on this extra
  information since they do not have access to this information
  in~$\game$ either.  We thus interpret $\bothgameadmp{\star}{p}$ as a
  game of \emph{imperfect information} where all the states labelled by
  the same state of $\game$ are in the same observation class. Thanks
  to these constructions, we can finally encode the fact that the
  players must play $\star$-SCO strategies in the new objective of the
  games, which will be given as an LTL formula, as we describe below.
\end{inparaenum}

To simplify the presentation, we restrict ourselves to prefix
independent winning conditions, although the results can be
generalised.  Also, to ensure we can effectively solve subproblems
mentioned above, we consider $\omega$-regular objectives.
The values of the histories
depend thus only on their last states, i.e. for all pairs of histories
$\hist_1$ and $\hist_2$: $\last{\hist_1}=\last{\hist_2}$ implies that
$\valh{\star}{\hist_1}{p}=\valh{\star}{\hist_2}{p}$. We thus denote by
$\valh{\star}{s}{p}$ the value $\valh{\star}{\hist}{p}$ of all
histories $\hist$ s.t. $\last{\hist}=s$.  
Last, we assume that a player cannot play the
same action from two different states, i.e. $\forall s_1\neq s_2$,
$\choicep{s_1}(p)\cap\choicep{s_2}(p)=\emptyset$. Thus, we say
that \emph{a given move $a$ is $\star$-LA}, meaning that $a$ is
$\star$-LA from all histories ending in the unique state where $a$ is available.

\subparagraph*{The game $\bothgameadmp{\star}{p}$} 
Let us now describe precisely the construction of
$\bothgameadmp{\star}{p}$.  Given an $n$-player concurrent game
$\game=(\states, \Sigma,\vinit, (\choicep{p})_{p\in\players},
\transfun)$ with winning condition $\Phi$ considered under semantics
$\star\in\{\sure,\asure\}$, and given a player $p$, we define the
two-player zero-sum concurrent game
$\bothgameadmp{\star}{p}=(\statesbis, \Sigmabis,\vinitbis,
(\choicepbis{p},\choicepbis{-p}),\transfunbis)$ where:
\begin{inparaenum}[(i)]
\item $\statesbis=\states\times\Sigmabis^n\cup\{\vinitbis\}$;
\item $\Sigmabis$ is the set of Dirac $\star$-LA moves in
  $\Sigma$;
\item $\vinitbis=\vinit$ is the initial state;
\item $\choicepbis{p}$ is such that $\choicepbis{p}(s)$ is the set of
  Dirac $\star$-LA moves of $p$ in $s$, for all $s\in S$;
\item $\choicepbis{-p}$ is s.t. for all $s\in S$: $\choicepbis{-p}(s)$
  is the set of moves $\vec a$ of $-p$ in $s$ s.t. for all $q\neq p$,
  $a_q$ is a Dirac $\star$-LA move;
\item $\transfunbis$ updates the state according to $\transfun$,
  remembering the last actions played:
  $\transfunbis(\vinitbis,\vec b)=(\transfun(\vinit,\vec b),\vec b)$
  and $\transfunbis((s,\vec a),\vec b)=(\transfun(s,\vec b),\vec b)$
  for all $s\in S$.
\end{inparaenum}
Note that the game $\bothgameadmp{\star}{p}$ depends on whether
$\star=\asure$ or $\star=\sure$ because the two semantics yield
different sets of LA-moves.  However, we interpret
$\bothgameadmp{\star}{p}$ in the sure semantics, so both players can
play Dirac strategies only in $\bothgameadmp{\star}{p}$.

Let us now explain how we obtain an imperfect information game by
defining an \emph{observation function} $\obs$.  Note that histories
in $\bothgameadmp{\star}{p}$ are of the form:
$\histbis=\vinitbis (s_1,\vec a_1)(s_2,\vec a_2)\cdots(s_n,\vec
a_n)$. Then, let $\obs:\statesbis\to \states$ be the mapping that,
intuitively, projects moves away from states. For example, in
\figurename~\ref{fig:imperfectinfo}, states with observation $s_0$ are
in the dashed rectangle. That is: $\obs(s,\vec a)=s$ for all states
$s$, and $\obs(\vinitbis)=\vinit$.  We extend $\obs$ to histories
recursively: $\obs(\vinitbis)=\vinit$ and
$\obs(\hist(s_n,\vec a_n))=\obs(\hist) s_n$.  To make
$\bothgameadmp{\star}{p}$ a game of imperfect information, we request
that, in $\bothgameadmp{\star}{p}$, players play only strategies $\strat$
s.t. $\strat(\hist_1)=\strat(\hist_2)$ whenever
$\obs(\hist_1)=\obs(\hist_2)$.

We relate the strategies in the original game $\game$
with the strategies in $\bothgameadmp{\star}{p}$, which we need to
extract admissible strategies in $\game$ from the winning strategies
in $\bothgameadmp{\star}{p}$ and thus perform assume-admissible
synthesis.
First, given a player $p$ strategy $\strat$ in $\game$ (i.e.,
$\strat\in\rectset_p(\game)$), we say that a strategy
$\overline{\strat}\in\rectsetdet_p(\bothgameadmp{\star}{p})$ is a
\emph{realisation} of $\strat$ iff:
\begin{inparaenum}[(i)]
\item $\overline{\strat}$ is Dirac; and
\item $\overline{\strat}(\hist)\in\supp(\strat(\hist)))$ for all $\hist$.
\end{inparaenum}
Note that every $\star$-LA strategy $\strat\in\rectset_i(\game)$
admits realisations $\strat$ in $\rectset_i(\bothgameadmp{\star}{p})$.
Second, given a player $p$ Dirac strategy $\strat$ in
$\bothgameadmp{\star}{p}$ (i.e.,
$\strat\in\rectsetdet_p(\bothgameadmp{\star}{p})$) we say that
$\hat{\strat}\in\rectset_p(\game)$ is an \emph{extension} of $\strat$
iff, for all $\hist\in\prefixoutc(\bothgameadmp{\star}{p},\strat)$:
$\hat{\strat}(\obs(\hist))=\strat(\hist)$.

\subparagraph*{The assume-admissible synthesis technique} 
As explained above, the assume-admissible rule boils down to computing
a winning strategy $\stratbis$ for player $p$ in
$\bothgameadmp{\star}{p}$ w.r.t. to winning condition
$\bothPhiAA{\star}{p}$, and extracting, from $\stratbis$, the required
admissible strategy in $\game$.

We will now formally define $\bothPhiAA{\star}{p}$.
Let $p$ be a player (in $\game$); and let us
denote by $\stateofaction(a)$ the (unique) state from which $a$ is available, for all actions~$a$.
We define $\bothAfterHelpMove{\star}{p}$ as
\begin{align*}
  \bothAfterHelpMove{\star}{p} &=\{(s,\vec a)\in\statesbis \mid \exists
                                 s'\in\succmovenew{\stateofaction(a_p),a_p}:
                                 \valh{\star}{s'}{p}\geq 0 
                                 \wedge s'\neq
                                 s \wedge \valh{\star}{s}{p}=0\}.
\end{align*}
That is, when $ (s,\vec a)\in \bothAfterHelpMove{\star}{p}$, in
$\game$, player $p$ has played $a_p$ from $\stateofaction(a_p)$ and,
due to player $-p$'s choice, $\game$ has reached $s$. However, with
another choice of player $-p$, the game could have moved to a
different state $s'$ from which $-p$ can help $p$ to win as
$\valh{\star}{s'}{p}\geq 0$. Intuitively, in runs that visit states of
value $0$ infinitely often, states from $\bothAfterHelpMove{\star}{p}$
should be visited infinitely often for player $i$ to play SCO,
i.e. such runs might not be winning, but this cannot be blamed on
player $p$ who has sought repeatedly the collaboration of the other
players to enforce his objective.  Observe further that the definition
of this predicate requires the labelling of the states (by actions) we
have introduced in $\bothgameadmp{\star}{p}$. For example, in
\figurename~\ref{fig:imperfectinfo},
$\bothAfterHelpMove{\asure}{2}=\big\{\big(s_0,(a,b')\big)
,\big(s_1,(b,b')\big)\big\}$.
We let
$\bothphivalzero{\star}{p}=\Diamond
\neg\bothvalpredstates{\star}{p}{0}\vee \win{p}\vee \Box \Diamond
\bothAfterHelpMove{\star}{p}$ and
$\bothphivalone{\star}{p}=\left(\Diamond
  \bothvalpredstates{\star}{p}{1}\right)\rightarrow \win{p}$. 
Let us define 
\(
\bothPhiAA{\star}{p} =
\left(\bigwedge_{q\neq p}
    \bothphivalzero{\star}{q}\wedge\bothphivalone{\star}{q}\right)\rightarrow
  (\win{p}\vee \Diamond\bothvalpredstates{\star}{p}{1}).
\)

To establish the correctness of this approach, we need several
auxiliary lemmas.
\begin{lemma}\label{lem:admtoAAAS}
  If a strategy $\strat\in\rectset_p(\game)$ is $\star$-admissible
  then for all of its realisations
  $\overline{\strat}\in\rectsetdet_p(\bothgameadmp{\star}{p})$,
  $\bothgameadmp{\star}{p},\overline{\strat}\smodels{\bothphivalzero{\star}{p}}$.
\end{lemma} 
\begin{proof}
  Assume toward a contradiction that there is a $\star$-admissible
  strategy $\strat\in\rectset_p(\game)$ for which there exists a
  realisation
  $\overline{\strat}\in\rectsetdet_p(\bothgameadmp{\star}{p})$ that
  does not satisfy $\bothphivalzero{\star}{p}$.  We can further assume
  wlog.~that $\strat$ plays Dirac moves in states of value $0$.  There
  is a run
  $\rho\in \outcome{\bothgameadmp{\star}{p},\overline{\strat}}$ that
  does not satisfy $\bothphivalzero{\star}{p}$. More precisely, $\rho$
  satisfies
  \[
    \Box\bothvalpredstates{\star}{i}{0}\wedge \neg \win{i}\wedge
    \Diamond \Box \neg\bothAfterHelpMove{\star}{p}.
  \]
  The run $\rho$ does not satisfy $\Phi(p)$ but all values of its
  prefixes are equal to $0$ and there is $k_0$ such that for every
  $k\geq k_0$, $\rho_k \not\in \bothAfterHelpMove{\star}{p}$. We now
  show that for every $\rho'\in \outcome{\game,\strat}$, such that
  $\obs(\rho_{\leq k_0})\prefix \rho'$, it holds that
  $\rho'\not \in \Phi(p)$. If $\rho'=\obs(\rho)$ then
  $\rho'\not \in \Phi(p)$. If $\rho'\neq\obs(\rho)$ then there exists
  $k\geq k_0$ for which, $\rho'_{\leq k} =\obs(\rho)_{\leq k}$ and
  $\rho'_{\leq k+1} \neq\obs(\rho)_{\leq k+1}$. Let
  $\hist'= \rho'_{\leq k} =\obs(\rho)_{\leq k}$, $s=\last{\hist'}$,
  $a=\overline{\strat}(\rho_{\leq k})$. $\hist'$ as a prefix of
  $\obs(\rho)$ has value $0$ so $\strat(\hist')$ is Dirac and thus
  equal to $a$.  
  Let
  $c,c'\in\choicep{-p}(s)$ such that
  $\transfun(s,(a,c))=\obs(\rho)_{k+1}$,
  $\transfun(s,(a,c'))=\rho'_{k+1}$.  From $\obs(\rho)_{k+1} \not \in
  \bothAfterHelpMove{\star}{p}$ we know that $\valh{\star}{\rho'_{\leq
      k+1}}{p}=\valh{\star}{\hist'\transfun(s,(a,c'))}{p}=-1$.  We
  have found a prefix of $\rho'$ that has value $-1$, thus $\rho'\not
  \in\Phi(p)$. We have proved that for all run $\rho'\in
  \outcome{\game,\strat}$ such that $\obs(\hist)\prefix
  \rho'$, it holds that $\rho'\not \in
  \Phi(p)$. This implies that
  $\valhstrat{\star}{\obs(\hist)}{\strat}=-1< 0=
  \valh{\star}{\obs(\hist)}{p}$ proving that
  $\strat$ is not $\star$-SCO and hence not
  $\star$-admissible. This is a contradiction.
\end{proof}

\begin{lemma}\label{lem:AAAStoadm} Given a $\star$-admissible strategy
  $\strat\in
  \rectset_p(\game)$ 
  and $\overline{\strat}\in\rectsetdet_p(\bothgameadmp{\star}{p})$ one
  of its realisation. If a run
  $\rho\in\outcome{\bothgameadmp{\star}{p},\overline{\strat}}$
  satisfies
  $\Box (\neg \bothvalpredstates{\star}{p}{1})\wedge\bigwedge_{q\neq
    p} \bothphivalzero{\star}{q}\wedge\bothphivalone{\star}{q}$ in
  $\bothgameadmp{\star}{p}$ then there exists a profile
  $\tau\in\rectset_{-p}(\game)$ of $\star$-admissible strategies such
  that $\outcome{\game,(\strat,\tau)}=\{\obs(\rho)\}$. \end{lemma}

\begin{proof}
  For every $k\geq 1$ there exists moves
  $a^k_1,\ldots,a^k_n\in \Sigmabis$ such that
  $\transfun(\rho_k, a^k_1,\ldots,a^k_n)=\rho_{k+1}$ and
  $a^k_p=\overline{\strat}(\rho_{\leq
    k})$. 
  For every $q\neq p$, let $\tau_q\in\rectset_{q}(\game)$ be a
  strategy such that for every $k\geq 1$,
  $\tau_q(\obs(\rho_{\leq k}))=a^k_q$ and that plays an admissible
  strategy as soon as the current history is not a prefix of
  $\obs(\rho)$. We now show that
  $\outcome{\game,(\strat,\tau)}=\{\obs(\rho)\}$. Denote by
  $a^k_{-p}=(a^k_q)_{q\in\players\setminus\{p\}}$ the Dirac move
  played by $-p$ in $\obs(\rho_{\leq k})$. Since $\strat$ is
  $\star$-admissible; for every $k\ge 1$, $\strat(\rho_{\leq k})$ is
  $\star$-LA and by Lemma \ref{lem:characadmmoves}
  (\ref{it3:characadmmoves}) equivalent to the move of its support
  $a^k_p$, then by Lemma \ref{lem:samedistronestep}, we get
  $\transfunrand(\obs(\rho_{k}),(\strat(\obs(\rho_{\leq
    k})),a^k_{-p}))(\obs(\rho_{k+1}))=\transfunrand(\obs(\rho_{k}),(a^k_p,a^k_{-p}))(\obs(\rho_{k+1}))=1$.
  This implies that $\outcome{\game,(\strat,\tau)}=\{\obs(\rho)\}$.

  We now show that $\tau_q$ is admissible in $\game$ for every
  $q\neq p$.
  It is clearly $\star$-LA as the moves $a^k_q$ are $\star$-LAs and
  $\tau_q$ plays $\star$-LA moves at histories not prefix of
  $\obs(\rho)$.
  It is also $\star$-SCO at histories not prefix of $\obs(\rho)$. We
  now show that $\tau_q$ is $\star$-SCO at every
  $\hist\prefix\obs(\rho)$. 

  If $\valh{\star}{\hist}{q}=1$ then by $\bothphivalone{\star}{q}$,
  $\rho$ is winning for $\Phi(q)$, and hence so is $\obs(\rho)$. As
  $a^k_q$ is a $\star$-LA move for every $k$, all histories of the
  form $\obs(\rho_{\leq k})$ for $k\geq |\hist|$ and their successors
  through $a^k_q$ have value $1$. 
  Consider a profile $\profile$ such that $\profile_q=\tau_q$,
  $\hist\in \prefixoutc(\profile)$, we show that 
  $\game,\profile \starmodelsh{\hist}{\Phi(q)}$. We do this part of proof 
  for $\star=\asure$ (the proof for $\star=\sure$ is a bit easier and omitted).
  We show that $\Prob_{\profile}(\neg \Phi(q) \cap \cyl{\hist})=0$. 
  For this, we decompose the set $\neg \Phi(q) \cap \cyl{\hist}$ into disjoint sets of runs: 
  \begin{equation}\label{eq:bigcup}
  \neg \Phi(q)\cap \cyl{\hist}= (\neg \Phi(q)\cap \{\obs(\rho)\})\cup  \bigcup_{k\geq |\hist|}\bigcup_{s'}\left(\neg \Phi(q)\cap \cyl{\obs(\rho_{\leq k})s'}\right)
  \end{equation}
  where the last union range over the state
  $s'\in\supp(\transfunrand(\rho_{k},\profile(\obs(\rho_{\leq
    k})))\setminus\{\obs(\rho_{k})\}$.
  The histories $\obs(\rho_{\leq k}) s'$ quantified above have value
  $1$, and $\tau_q$ is $\asure$-admissible at these histories by
  construction since the histories are not prefixes
  of~$\obs(\rho)$. Hence, $\asure$-winning from these histories, so
  $\Prob_{\profile}(\neg \Phi(q) \cap \cyl{\obs(\rho_{\leq k})
    s'})=0$.  $\Prob_{\profile}(\neg \Phi(q) \cap \{\obs(\rho)\})=0$
  because $\obs(\rho)\in\Phi(q)$.  The set
  $\neg \Phi(q)\cap \cyl{\hist}$ is a countable union of sets of
  probability $0$, it is thus also of probability $0$.  The proof for
  the sure semantics is similar: it suffices to replace the fact that
  a set has probability zero by emptiness of this set.  We have thus
  proved that
  $\valhstrat{\star}{\hist}{\strat}=\valh{\star}{\hist}{q}=1$.

  If $\valh{\star}{\hist}{q}=0$. We show that there exists a profile
  $\profile$ such that $\profile_q=\tau_q$,
  $\hist\in\prefixoutcome{\game,\profile}$ and
  $\profile \starmodels{\Phi(q)}$. Let $k_0=|\hist|$. If
  $\obs(\rho)\in\Phi(q)$ then it suffices to take
  $\profile=(\strat,\tau)$. Otherwise we begin to show that there is
  $k\geq k_0$ such that $\rho_{k+1}\in\bothAfterHelpMove{\star}{q}$.
  Assume toward a contradiction that such a $k$ does not exists, then
  $\bothphivalzero{\star}{q}$ tells us that there exists
  $k_2$ 
  such that $\rho_{k_2} \in \bothvalpredstates{\star}{q}{-1}$ or
  $\rho_{k_2} \in \bothvalpredstates{\star}{q}{1}$. Necessarily
  $k_2>k_0$ because as $\star$-LA moves only are played the states of
  value $1$ cannot be escaped once entered. In the case
  $\rho_{k_2} \in \bothvalpredstates{\star}{q}{1}$, $\rho$ satisfies
  $\Diamond \bothvalpredstates{\star}{q}{1}$ and hence $\Phi(q)$ by
  virtue of $\bothphivalone{\star}{q}$. This contradict the fact that
  we are considering the case $\obs(\rho)\not\in\Phi(q)$. Consider now
  that $\rho_{k_2} \in \bothvalpredstates{\star}{q}{-1}$. Let $k_1$ be
  the unique number in $\{k_0,\ldots,k_2-1\}$ such
  that 
  $\valh{\star}{\rho_{k_1}}{q}=0$ and
  $\valh{\star}{\rho_{k_1+1}}{q}=-1$. Since $k_1\geq k_0$,
  $\rho_{k_1+1}\not\in\bothAfterHelpMove{\star}{q}$, from which we
  deduce that
  $\succmove{\rho_{k_1}}{a_q^k}\subseteq
  \bothvalpredstates{\star}{q}{-1}$. This contradicts the fact that
  $a_q^k$ is a $\star$-LA move in a history of value $0$. We have
  showed the existence of a prefix length $k\geq |\hist|$ such that
  $\rho_{k+1}\in\bothAfterHelp{\star}{\rho_k}{q}$ and we resume the
  proof of
  $\valhstrat{\star}{\hist}{\strat}=0=\valh{\star}{\hist}{q}$. There
  exists a state
  $s''\in\succmovenew{a^{k+1}_q}\setminus\{\obs(\rho_{k+1})\}$ such
  that $\valh{\star}{s''}{q}\geq 0$. Moreover $\tau_q$ is
  $\star$-admissible from $\obs(\rho_{\leq k}) s''$, so there is a
  profile $\profile$ such that
  $\obs(\rho_{\leq k}) s''\in\prefixoutcome{\profile}$,
  $\profile_q=\tau_q$ and $\game,\profile\starmodels{\Phi(q)}$. Since
  $\hist\prefix \obs(\rho_{\leq k}) s''$, this implies that
  $\valhstrat{\star}{\hist}{\strat}\geq 0$ as required.

  It remains the case of histories $\hist$ prefix of $\obs(\rho)$ such
  that $\valh{\star}{\hist}{q}=-1$. This case is
  straightforward 
  because the following sequence of inequalities holds for every
  history
  $-1\leq\valhstrat{\star}{\hist}{\strat}\leq\valh{\star}{\hist}{q}$.
  At the end, for every $q\neq p$, $\tau_q$ is $\star$-admissible as
  both $\star$-LA and $\star$-SCO, and
  $\obs(\rho)=\outcome{\game,(\strat,\tau)}$. 
\end{proof}

\begin{lemma}[Admissible and dominant synthesis]\label{lem:admsynth}
  In concurrent games with prefix independent $\omega$-regular winning
  conditions with semantics $\star\in\{\sure,\asure\}$, given a finite
  memory strategy $\strat$ for player $p$ that is $\star$-LA, one can
  construct effectively a finite memory strategy $\strat'$ such that
  $\strat\bothdomstrateq{\star}\strat'$.
\end{lemma}\label{lemma:build-strat}
\begin{proof}
  We assume that the strategy $\strat$ is given by a stochastic Moore
  machine ${\cal M}_{\strat}$ with set of memory states
  $M=\{ m_1,m_2, \dots , m_k \}$, with a designated initial memory
  state $m_{\sf init}$ and such that after an history $h$, the machine
  is in memory state ${\cal M}_{\strat}(h) \in M$. Each such memory
  state is labelled with a move noted ${\cal L}_{\strat}(m)$. If we
  take the product of ${\cal M}_{\strat}$ with the game $\game$, we
  obtain a structure with states that are pairs $(s,m)$ where $s$ is a
  state of the original game and $m$ is a memory state of
  ${\cal M}_{\strat}$.  This product can be seen as a new Moore
  machine that behaves exactly as ${\cal M}_{\strat}$. From now on, we
  make the hypothesis that ${\cal M}_{\strat}$ has this state space.

  We note $\valh{\star}{s}{p}$ the value of state $s$ in $\game$, this
  value is well defined as $\Phi(p)$ is prefix independent, and we
  note $\valh{\star}{s,m}{\sigma}$ the value of the strategy $\sigma$
  for any history $h$ that ends up in state $s$ of the $\game$ and
  memory state $m={\cal M}_{\strat}(h)$ of the machine.

  As the strategy $\sigma$ is $\star$-LA, if it is not admissible then
  it must be the case that there are state $s$, memory state $m$, and
  histories $h$ such that after history $h$, the game $\game$ is in
  state $s$, ${\cal M}_{\sigma}$ is in state $m$, and
  $\valh{\star}{s}{p} \not= \valh{\star}{s,m}{\sigma}$, i.e. $h$ is a
  witness that shows that $\sigma$ is not SCO. Then, we will modify
  systematically ${\cal M}_{\strat}$ in a way that it behaves as
  ${\cal M}_{\sigma}$ in all other histories and plays SCO in the
  histories that witness the fact that $\sigma$ is not SCO.

  We only need to consider the two following cases:
  \begin{enumerate}
  \item $s$ and $(s,m)$ are such that: $\valh{\star}{s}{p}=1$ and
    $\valh{\star}{s,m}{\sigma} \leq 0$, then we replace $(s,m)$ by a
    sub-machine that implements a $\star$-winning strategy from
    $(s,m)$. We know that finite memory machines that are
    $\star$-winning always exist and can be computed effectively for
    all $\omega$-regular objectives.
  \item $s$ and $(s,m)$ are such that: $\valh{\star}{s}{p}=0$ and
    $\valh{\star}{s,m}{\sigma}=-1$, then we replace $(s,m)$ by a
    sub-machine that plays any admissible strategy from $s$. We show
    this by establishing in the next claim that there is always such a
    finite state strategy and it can be computed effectively.
   \end{enumerate}

   \noindent {\bf Claim}. Given a game $\game$, a state $s$, and a
   player $p$, we can construct a finite state stochastic Moore
   machine that encodes an admissible strategy.

   We establish this claim constructively. We consider the following
   case study to describe the machine:
  \begin{itemize}
  \item if the value of $\valh{\star}{s}{p}=1$, then the machine plays
    a finite state $\star$-winning strategy, such a finite memory
    strategy always exists and can be computed effectively.
  \item if the value of $\valh{\star}{s}{p}=-1$, then the machine
    plays arbitrarily.
  \item if the value of $\valh{\star}{s}{p}=0$, then the machine
    selects a finite lasso-shape path
    $\rho=\rho_1 \cdot \rho^{\omega}_2$ such that $\rho$ is compatible
    with $\star$-LA moves of player $p$ and such that
    $\rho \in \Phi(p)$. Such a finite lasso-path always exists as
    $\valh{\star}{s}{p}=0$. Then the machine plays according to this
    lasso-path either forever or up to a deviation by another
    player. If the lasso path is played forever then the outcome is
    $\rho \in \Phi(p)$, and if there is a deviation, the new state of
    the game is $s'$ and the three rules here are applied from $s'$
    (according to the value of state $s'$).
  \end{itemize}
  Clearly, if entering a state $s$ with value $0$, we always choose
  the same finite lasso-path then the machine has finite state. As in
  states with value $0$, it always plays $\star$-LA moves and $\rho$
  from $s$ is such that $\rho \in \Phi(p)$, then we conclude that the
  strategy is $\star$-SCO and $\star$-LA and thus
  $\star$-admissible. So, we are done.
\end{proof}

\begin{theorem}[Assume-admissible synthesis]\label{theo:aasynth}
  Player $p$ has a $\star$-admissible strategy $\strat$ that is
  $\star$-winning against all player $-p$ $\star$-admissible
  strategies in $\game$ iff Player $p$ has an $\sure$-winning strategy
  in $\bothgameadmp{\star}{p}$ for the objective
  $\bothPhiAA{\star}{p}$. Such a $\star$-admissible strategy $\strat$
  can be effectively computed (from any player $p$ $\sure$-winning
  strategy in $\bothgameadmp{\star}{p}$).
\end{theorem}
\begin{proof}
  \noindent\textbf{Completeness:}
  Assume there exists an admissible strategy
  $\strat\in\rectset_p(\game)$ that wins against every admissible
  strategy.  Let $\overline{\strat}\in\rectsetdet_p(\game)$ be a
  realisation of $\strat$.  
  Note
  that every runs
  $\rho\in\outcome{\bothgameadmp{\star}{p},\overline{\strat}}$ that
  satisfies $\Diamond \bothvalpredstates{\star}{p}{1}$ also satisfies
  $\bothPhiAA{\star}{p}$.  The other runs
  $\rho\in\outcome{\bothgameadmp{\star}{p},\overline{\strat}}$ satisfy
  $\neg\Diamond \bothvalpredstates{\star}{p}{1}\equiv \Box (\neg
  \bothvalpredstates{\star}{p}{1})$; for these runs we show that
  $\bigwedge_{q\neq p}
  \bothphivalzero{\star}{q}\wedge\bothphivalone{\star}{q}\rightarrow
  \win{p}$. So let
  $\rho\in\outcome{\bothgameadmp{\star}{p},\overline{\strat}}$ be such
  that
  $ \Box (\neg \bothvalpredstates{\star}{p}{1}) \wedge
  \bigwedge_{q\neq p}
  \bothphivalzero{\star}{q}\wedge\bothphivalone{\star}{q}$.  By Lemma
  \ref{lem:AAAStoadm}, there exists $\tau$ that only contains
  admissible profiles such that
  $\outcome{\strat,\tau}=\{\obs(\rho)\}$.  By assumption
  $(\strat,\tau)$ is winning for $\Phi(p)$, hence so is $\rho$.
  
  \noindent\textbf{Correctness:}
  We show the following stronger statement.  If a strategy $\strat$ is
  winning for $\bothPhiAA{\star}{p}$ in $\bothgameadmp{\star}{p}$ then
  one can construct a strategy $\hat \strat$ admissible that wins
  against every profile of admissible strategies as follows:
  \begin{inparaenum}[(i)]
  \item take an extension $\strat'$ of $\strat$; 
  \item modify $\strat'$ into $\strat''$ that plays a winning strategy
    as soon as a history of value $1$ is entered;
  \item use Lemma~\ref{lemma:build-strat} to design an admissible
    strategy $\hat \strat$ that weakly dominates $\strat''$.
  \end{inparaenum}

  We do the proof for the almost-sure semantics only, that is, for
  $\star=\asure$. The slightly simpler but similar proof for
  $\star=\sure$ is
  omitted. 
  We first show that if a strategy $\strat$ is winning for
  $\bothPhiAA{\asure}{p}$ in $\bothgameadmp{\asure}{p}$ then for every
  extension $\strat'$ of $\strat$ and every profile $\tau$ of
  admissible strategies, it holds that
  $\game,(\strat',\tau)\rmodels{\win{p}\vee
    \Diamond\bothvalpredstates{\asure}{p}{1}}$.

  Let $\tau$ be a profile of admissible strategies. Note that
  $\Prob_{(\strat',\tau)}(\Diamond
  \bothvalpredstates{\asure}{p}{1}\vee \Phi(p))=1$ is equivalent to
  $\Prob_{(\strat',\tau)}(\Box (\neg
  \bothvalpredstates{\asure}{p}{1})\wedge \neg \Phi(p))=0$.  For
  $q\neq p$, $\tau_q$ is admissible, so, winning from histories of
  value $1$, hence
  $\Prob_{(\strat',\tau)}(\bothphivalone{\asure}{q})=1$ and we have
  $\Prob_{(\strat',\tau)}(\Box (\neg
  \bothvalpredstates{\asure}{p}{1})\wedge \neg
  \Phi(p))=\Prob_{(\strat',\tau)}( \bigwedge_{q\neq p}
  \bothphivalone{\asure}{q}\wedge\Box (\neg
  \bothvalpredstates{\asure}{p}{1})\wedge \neg \Phi(p) )$.  To show
  that this probability is null it suffices to show that every run
  $\rho\in \outcome{\game,(\strat',\tau)}$ that satisfies
  $\bigwedge_{q\neq p} \bothphivalone{\asure}{q}\wedge\Box (\neg
  \bothvalpredstates{\asure}{p}{1})$ also satisfies $\Phi(p)$.  Take
  such a run $\rho$, then there is a sequence of move profile
  $\vec a^k$ such that $a^k_q\in \supp(\tau_q(\rho_k))$ and
  $\rho_{k+1}=\transfun(\rho_{k}, \vec a^k)$.  Define the run
  $\overline{\rho}$ by $\overline{\rho}_k=(\rho_{k}, \vec a^k)$ for
  $k\geq 1$ and $\overline{\rho}_0=\vinit'$ and note that
  $\rho=\obs(\overline{\rho})$.  We also have that $\overline{\rho}$
  is a run of $\outcome{\bothgameadmp{\asure}{p},\strat}$ and for
  $q\neq p$ of $\outcome{\bothgameadmp{\asure}{p},\overline{\tau_q}}$
  where $\overline{\tau_q}$ is any realisation of $\tau_q$ for which
  $a^k_q=\tau_q(\rho_k)$ for every $k\geq 1$.  
  The run $\overline{\rho}$ satisfies
  $\bigwedge_{q\neq p} \bothphivalone{\asure}{q}$ by assumption on
  $\rho$ and $\bigwedge_{q\neq p} \bothphivalzero{\asure}{q}$ because
  all the $\tau_q$ are admissible (Lemma \ref{lem:admtoAAAS}).
  Putting all the things together $\overline{\rho}$ satisfies
  $\Phi(p)$ and so does $\rho$.  We conclude that $\strat'$ satisfies
  that for every profile $\tau$ of admissible strategies, it holds
  that
  $(\strat',\tau)\rmodels{\win{p}\vee
    \Diamond\bothvalpredstates{\asure}{p}{1}}$.  By construction
  $\strat''$ wins against every profile $\tau$ for which
  $(\strat',\tau)\rmodels{\Diamond\bothvalpredstates{\asure}{p}{1}}$.
  So $\strat''$ is $\asure$-winning against every profile $\tau$ of
  $\asure$-admissible profile.  Since
  $\strat''\bothdomstrateq{\asure} \hat{\strat}$, we deduce that
  $\hat{\strat}$ is $\asure$-winning against every $\asure$-admissible
  profile.
\end{proof}

Let us explain how we build a strategy in $\game$ with the desired
properties, from any player~$p$ strategy enforcing
$\bothPhiAA{\star}{p}$ in $\bothgameadmp{\star}{p}$. Remember that
$\bothgameadmp{\star}{p}$ ensures that the players play $\star$-LA
moves only. We will use $\bothPhiAA{\star}{p}$ to make sure that, when
SCO strategies are played by $-p$ (relying on the extra information we
have encoded in the states), then $p$ reaches a state of value $1$.
First, consider $\bothphivalzero{\star}{q}$ for $q\neq p$. Runs that
satisfy this formula are either those that visit states of value $0$
only finitely often ($\Diamond \neg\bothvalpredstates{\star}{q}{0}$);
or those that stay in states of value $0$, in which case they must be
either winning ($\win{q}$) or visit infinitely often states where
Player $q$ could have been helped by the other players
($\Box \Diamond \bothAfterHelpMove{\star}{q}$). This is a necessary
condition on runs visiting only value $0$ states for the strategy to
be SCO. Next, observe that $\bothphivalone{\star}{q}$ states that
\emph{if} a history of value $1$ is entered \emph{then} Player $q$
must win. This allows us to understand the left part of the
implication in $\bothPhiAA{\star}{p}$: the implication can be read as
`if all other players play a $\star$-admissible strategy, then either
$p$ should win ($\win{p}$) or a state of value $1$ for player $p$
should eventually be visited
($\Diamond\bothvalpredstates{\star}{p}{1}$)'.  Then a strategy
$\stratter$ (in $\game$) that wins against admissible strategies can
be extracted from a winning strategy $\stratbis$ (in
$\bothgameadmp{\star}{p}$) in a straightforward way, \emph{except}
when $\stratbis$ enforces to reach a state of value $1$
($\Diamond\bothvalpredstates{\star}{p}{1}$ in
$\bothPhiAA{\star}{p}$). In this case, $\strat$ cannot follow
$\stratbis$, but must rather switch to a \emph{winning} strategy,
which:
\begin{inparaenum}[(i)]
\item is guaranteed to exist since the state that has been reached has
  value $1$; and
\item can be computed using classical
  techniques~\cite{DBLP:journals/tocl/ChatterjeeAH11}.
\end{inparaenum}
The strategy $\stratter$ is not necessarily admissible but
by Theorem \ref{theo:exist-adm} (\ref{item:7}), 
there is an admissible strategy $\strat$ 
with $\stratter\bothdomstrateq{\star} \strat$.
By  weak domination, $\strat$ wins against more profiles
than $\stratter$, in particular, it wins against the profiles of admissible  
strategies of the other players.

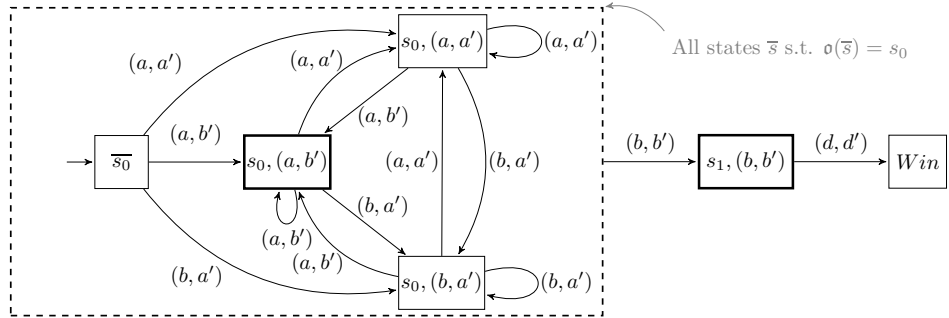
\begin{figure}[ht]
  \centering
\scalebox{0.8}{
  \begin{tikzpicture}[->,>=stealth',shorten >=1pt,auto,node distance=1.55cm, initial text={}]
    \node[state, rectangle,thick, dashed, inner sep=1pt] (q00) {
      \begin{tikzpicture}[solid]
        \node[initial,state,rectangle] (q1)                      {$\overline{s_0}$};
        \node[state,rectangle,very thick]          (q12) [right = of q1]         {$s_0,(a,b')$};
        \node[state,rectangle]          (q11) [above right  =of q12]         {$s_0,(a,a')$};
        \node[state,rectangle]          (q13) [below right =of q12]         {$s_0,(b,a')$};     
        \path (q1) edge  node {$(a,b')$} (q12);
        \path (q1) edge [bend right] node [left= 0.5 cm] {$(b,a')$} (q13); 
        \path (q1) edge [bend left] node[pos=0.2]  {$(a,a')$} (q11);      
        \path (q11) edge  node[outer sep=-5pt] {$(a,b')$} (q12);
        \path (q11) edge [loop right] node {$(a,a')$} (q11);
        \path (q12) edge [loop below] node[outer sep=-5pt] {$(a,b')$} (q12);
        \path (q13) edge [loop right] node {$(b,a')$} (q13);
        \path (q11) edge [bend left] node {$(b,a')$} (q13);
        \path (q12) edge [bend left] node[outer sep=-5pt] {$(a,a')$} (q11);
        \path (q12) edge  node[outer sep=-5pt] {$(b,a')$} (q13);
        \path (q13) edge  node {$(a,a')$} (q11);
        \path (q13) edge [bend left] node[outer sep=-5pt] {$(a,b')$} (q12);      
      \end{tikzpicture}
    };
    \node[state,rectangle,very thick]          (q2) [right= of q00]         {$s_1,(b,b')$};
    \node[state,rectangle]          (q3) [right=of q2]         {$Win$};   
    \path (q00) edge node {$(b,b')$} (q2);
    \path (q2) edge  node {$(d,d')$} (q3);   

    \node[above right  = -1cm and 1cm of q00, gray] (label) {All states $\overline{s}$
      s.t. $\obs(\overline{s})= s_0$} ;
    \path[gray] (label.north west) edge[bend right] (q00.north east) ;
  \end{tikzpicture}
}
\caption{\label{fig:imperfectinfo} The game $\bothgameadmp{\asure}{1}$
  obtained from the game in \figurename~\ref{ex:runex}. Bold states
  $\big(s_0,(a,b')\big)$ and $\big(s_1,(b,b')\big)$ are the states of
  $\bothAfterHelpMove{\asure}{2}$. There is a $(b,b')$-labelled
  transition from all states in the dashed rectangle to $\big(s_1,(b,b')\big)$.}
\end{figure}

\begin{example}
  In our running example, observe that
  $\neg\bothvalpredstates{\asure}{2}{0}=\bothvalpredstates{\asure}{2}{1}=\{Win\}$
  since there is no state of value $-1$ in $\game$. Hence,
  $\win{2}=\Diamond Win=\Diamond \bothvalpredstates{\asure}{2}{1} =
  \Diamond \neg\bothvalpredstates{\asure}{2}{0}$. Finally,
  $\bothAfterHelpMove{\asure}{2}=\big\{\big(s_0,(a,b')\big),\big(s_1,(b,b')\big)\big\}$,
  so, after simplification:
  $\bothPhiAA{\asure}{1}=\big[\Diamond Win \vee
  \Box\Diamond\big((s_0,(a,b')) \vee(s_1,(b,b'))\big)\big]\rightarrow
  \Diamond Win$. Thus, to win in $\bothgameadmp{\asure}{1}$ (under the
  \emph{sure} semantics), player 1 must ensure to reach $Win$ as long
  as player 2 visits the set of \emph{bold} states in
  \figurename~\ref{fig:imperfectinfo} infinitely often. A winning
  strategy $\stratbis$ in $\bothgameadmp{\asure}{1}$ consists in
  (eventually) always
    playing $b$ from all states in the
  dashed rectangle; and $d$ from $\big(s_1,(b,b')\big)$. Observe that
  this strategy is compatible with $\obs$. From $\stratbis$, we can
  extract an admissible player $1$ strategy in $\game$: always play
  $b$ in $s_0$; always play $d$ in $s_1$; and play a winning strategy
  from $s_2$ (which is of value $1$), for instance: always play
  $0.5f+0.5g$ from $s_2$ like $\strat_3$ does.
\end{example}

We conclude by a remark on games with simple safety objectives.

\begin{remark}
  In the case of simple safety games, the situation is much
  simpler. We have seen in Theorem~\ref{theo:classifconc} that, for
  simple safety objectives, $\star$-LA strategies are exactly the
  admissible strategies. So, can simply build $\game_p$ from $\game$
  by pruning the actions which are not $\star$-LA (the labelling by
  actions is not necessary anymore since its sole purpose is to
  enforce SCO), and look for a player $p$ winning strategy in the
  resulting game.
\end{remark}
